\documentclass[english,letterpaper,twocolumn]{IEEEtran}
\usepackage{tikz}
\usetikzlibrary{matrix}
\usetikzlibrary{matrix,plotmarks,arrows,snakes,backgrounds}
\usetikzlibrary{shapes,arrows,shadows}
\tikzstyle{vertex}=[auto=left,circle,fill=black!25,minimum size=20pt,inner sep=0pt]

\pdfoutput=1

 \usepackage{verbatim,multirow}
 \usepackage{amsmath}
 \usepackage{graphicx}
 \usepackage{amssymb}
\usepackage{amsthm}   
 \makeatletter
\usepackage{rotating}

\usepackage{booktabs}
\newcommand{\ra}[1]{\renewcommand{\arraystretch}{#1}}

\usepackage{gastex}

 \usepackage{cite}
 \usepackage{hhline}
 \usepackage{subfigure}

 \usepackage{array}
 \usepackage{acronym}
\usepackage{color}
\usepackage{algorithm}
\usepackage{algorithmic}
\usepackage{float}

\newtheorem{thm}{Theorem}
\newtheorem{cor}{Corollary}
\newtheorem{lem}{Lemma}
\theoremstyle{remark}

\theoremstyle{definition}
 \newenvironment{definition}[1][Definition]{\begin{trivlist}
 \item[\hskip \labelsep {\bfseries #1}]}{\end{trivlist}}

 \makeatletter \newcommand \listoftodos{\chapter{Todo list} \@starttoc{todocontents}}
 \newcommand\l@todo[2]
     {\par\noindent pg \textit{#2}: \parbox{10cm}{#1}\par} \makeatother

\newcommand{\etal}{\textit{et al.}~}

\acrodef{RV}{random variable}

\acrodef{i.i.d.}{independent, identically distributed}

\acrodef{PDF}{probability distribution function}
\acrodef{PMF}{probability mass function}
\acrodef{CDF}{cumulative distribution function}
\acrodef{HF}{high frequency}
\acrodef{ch.f.}{characteristic function}

\acrodef{NPV}{net present value}
\acrodef{PV}{present value}
\acrodef{CF}{cash flow}
\acrodef{DCF}{discounted cash flow}
\acrodef{WACC}{weighted average cost of capital}

\acrodef{CDN}{content distribution network}
\acrodef{NC}{network coding}
\acrodef{DC}{data center}
\acrodef{PUE}{power usage effectiveness}
\acrodef{RLC}{random linear coding}
\acrodef{ACPI}{advanced configuration and power interface}
\acrodef{RLNC}{random linear network coding}
\acrodef{SAN}{storage area network}
\acrodef{NAS}{network attached storage}
\acrodef{NCS}{network coded storage}
\acrodef{AWS}{Amazon Web Services}
\acrodef{HD}{high definition}
\acrodef{HLS}{HTTP Live Streaming}
\acrodef{NCC}{no coefficient-cycling}
\acrodef{CC}{coefficient-cycling}
\acrodef{HDD}{hard disk drive}
\acrodef{SSD}{solid state drive}
\acrodef{IDNC}{instantly decodable network coding}
\acrodef{SFM}{state feedback matrix}
\acrodef{DSM}{drive state matrix}
\acrodef{MDS}{maximum distance separable}
\acrodef{d.o.f.}{degree of freedom}
\acrodef{PMP}{point-to-multipoint}
\acrodef{QSS}{queueing-based storage switch}
\acrodef{SSP}{stable-set polytope}
\acrodef{QCN}{queued cross-bar network}
\acrodef{RR}{rate region}

\author{Ulric J.~Ferner, Neda Aboutorab, Parastoo Sadeghi, Muriel M\'{e}dard \thanks{This material is based upon
    work supported by 
the Martin Family Society of Fellows for Sustainability at MIT, by BAE Systems National Security
    Solutions Inc., under award 739532-SLIN 0004, and the Australian Research Councils Discovery Projects funding scheme (project no. DP120100160). U.J.~Ferner and M.~M\'{e}dard
    are currently with the Research Laboratory for Electronics,  Massachusetts Institute of Technology, Room 36-512, 77 Massachusetts Avenue, Cambridge,
    MA 02139 (e-mail: \{uferner, medard\}@mit.edu).  P.~Sadaghi and N.~Aboutorab are with the Research School of Information Sciences and Engineering, The Australian National University, Canberra, Australia (e-mail:\{parastoo.sadeghi, neda.aboutorab\}@anu.edu.au). }  }
\title{Queued cross-bar network models \\ for replication and coded storage systems}
\date{\today}

\begin{document}
\maketitle
\begin{abstract}

Coding techniques may be useful for data center data survivability as well as for reducing traffic
congestion.   We
present a \ac{QCN} method that can be used for traffic analysis of both replication/uncoded and
coded storage systems.   We develop a framework for generating \ac{QCN} \acp{RR} by analyzing their conflict graph stable set
polytopes (SSPs).  In doing so, we apply recent results from graph theory on the characterization of
particular graph SSPs.   We characterize the SSP  of \ac{QCN} conflict graphs under a variety of
  traffic patterns, allowing for their efficient \ac{RR} computation. For uncoded systems, we show
  how to compute \acp{RR} and find rate optimal scheduling algorithms.  For coded storage, we
  develop a \ac{RR} upper bound, for which we provide an intuitive interpretation.  We show that the coded storage \ac{RR} upper bound is
achievable in certain coded systems in which drives store sufficient coded information, as well in certain dynamic coding
systems.   Numerical illustrations show that coded storage can result in gains in \ac{RR} volume of
approximately 50\%, averaged across traffic patterns.  
\end{abstract}

\vspace{-5mm}
\section{Introduction}
\label{sec:introduction}

The continued growth in \ac{DC} demand worldwide is driving the development of new \ac{DC}
architectures and data management techniques.  Two key parameters of data management in
\acp{DC} are the survivability of data in the event of node failures and
constant availability of data.  Significant academic literature has
focused on data survivability in the form of regenerating codes.  See \cite{DimRamWuSuh:11} and
references therein.  In enterprise-level \acp{DC} temporary drive unavailability dominates permanent failure by a factor of
nine  to one \cite{ForLabPopStoTruBarGriQui:10}.   In this paper we concentrate on data unavailability
due to traffic congestion.  

We consider physical storage networks as illustrated in Fig.~\ref{fig:physicalNetEx1}.  Chunks are fixed-size file
subsets.  A network of drives, where each drive is either a
\ac{HDD}, \ac{SSD}, or RAM cache, stores some number of file chunks.   Outside users send read requests for 
file chunks to the drive network, and drives process read requests and send
chunks back to users.  Time is slotted and in each timeslot the network is constrained as to how
many users each drive can transmit a stored chunk to; we refer to these constraints as traffic patterns.

\begin{figure}[t]
  \centering
  \includegraphics[width=\linewidth]{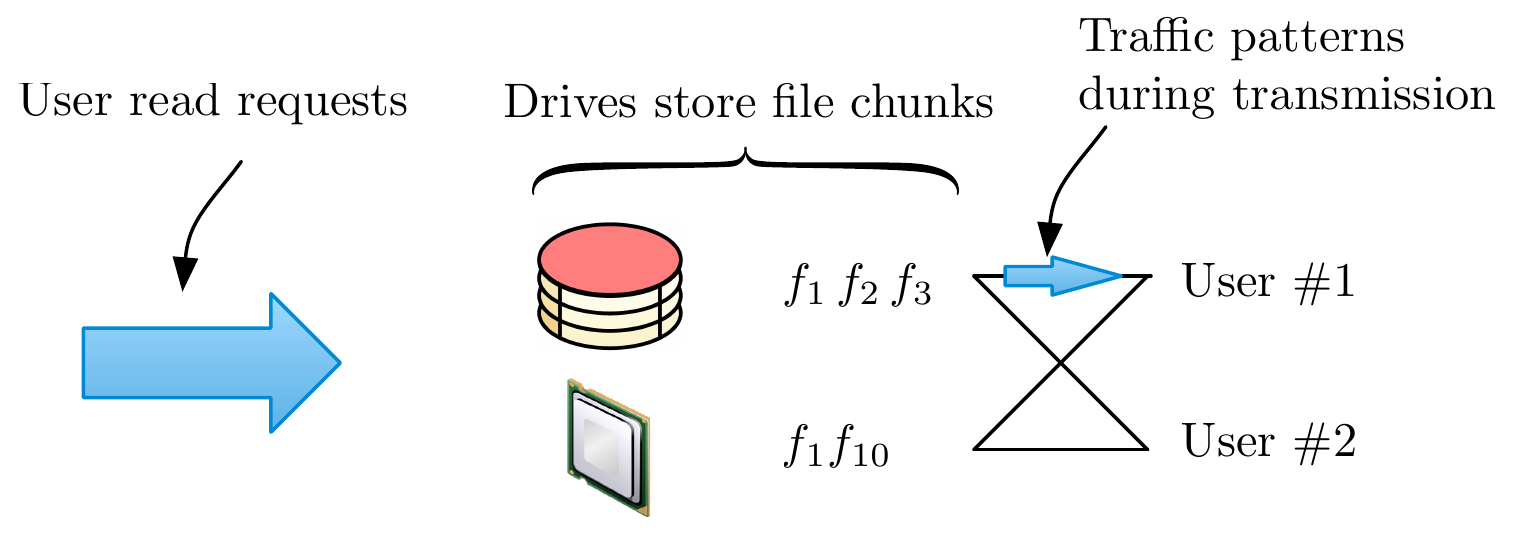}
  \caption{Illustration of the physical networks we consider in this paper.  User requests arrive
    at the drive network.  Drives, of various storage technologies, store some set of file chunks.
    Time is slotted and during each timeslot, traffic pattern constraints govern how many users each
    drive can transmit a stored chunk to.}
  \label{fig:physicalNetEx1}
\end{figure}

To the author's knowledge, the following key questions are unaddressed when designing and
implementing high-traffic storage networks:  What is the maximum achievable rate of a storage network with
general traffic patterns and arbitrary chunk-to-drive mappings?  There do not exist systematic
methods of mapping physical networks to queueing models.   What is the impact of
coded storage on the maximum achievable rate of a storage network?  Existing queueing work on
coded storage either assumes perfect scheduling or uses scheduling heuristics.  Further, what
scheduling algorithms achieve maximum rate for coded storage?  

Referring to Fig.~\ref{fig:method}, the main contributions of
this paper are as follows.
\begin{itemize}
      \item We introduce the \textit{\acf{QCN}} method, which is a technique
  to model relatively general physical drive
  networks as queueing networks.  Our technique allows for arbitrary traffic
  patterns as well as chunk-to-drive mappings.  
      \item To determine the effect of drive traffic pattern restrictions, we develop a framework
  for analyzing \ac{QCN} \acp{RR} by considering their
  conflict graph stable set polytopes (SSPs).  In
  doing so, we use and adapt existing techniques from cross-bar switching literature.
      \item We exactly characterize the SSP  of \ac{QCN} conflict graphs under a variety of
  traffic patterns, allowing for their efficient \ac{RR} computation.
      \item For uncoded storage, we characterize \acp{RR}, and prove
that the existing scheduling algorithm of Tassiulas \etal \cite{TasEph:92} can be modified to be
rate optimal in this application.  For coded storage, we develop a \ac{RR} upper bound, for which we
provide an intuitive interpretation.  We show that the coded storage \ac{RR} upper bound is
achievable in certain coded systems in which drives store sufficient coded information, as well in certain dynamic coding
systems.  
      \item We present numerical illustrations that show potential increases in \ac{RR} volume from
  coded storage,  averaging 50\% across traffic patterns.
\end{itemize}

This paper builds upon and complements existing work in coded storage.  General
scheduling for coded storage in point-to-point networks, when users are served sequentially instead 
of simultaneously, and with particular file layouts are considered in \cite{ShaLeeRam:12,HuaPawHasRam:12}.  
Server scheduling is also well studied in matched networks such as cross-bar switches.  Throughput-optimal
schedules are considered for $ N\times M$ point-to-point cross-bar switches using graph theory and
techniques such as the Birkhoff-von Neumann theorem \cite{AndOwiSaxTha:93,McKMekAnaWal:99,CarRosGoeTar:04}.  
Switches with multicast and broadcast capabilities with
a queueing analysis flavor are considered in \cite{MarBiaGiaLeoNer:03}.  References \cite{PraMcKAhu:97,YuReuBer:11} attempt to map the
multicast problem in cross-bar switches to simpler problems such as block-packing
games and round-robin based multicast.   Chunk scheduling problems in uncoded peer-to-peer
networks, as opposed to \acp{PMP}, are considered in \cite{FenLi:11}, and for star-based broadcast
networks in \cite{RouSiv:97}.   This manuscript differs from these works in
that we consider general traffic patterns and arbitrary chunk-to-drive mappings.  We
also consider scheduling for  coded storage.   

Scheduling for network coded multicast in multihop wireless networks, using a conflict graph approach, is considered in
\cite{TraHeiMedKoe:12}.  Reference \cite{KimSunMedEryKot:11} developed optimal scheduling algorithms for cross-bar switches
with network coding in communications.  Enhanced conflict graphs, closely related to classical conflict graphs, were
developed to allow analysis of cross-bar switches with network coding.  The stable sets of these
conflict graphs were then characterized exactly by showing that in certain cases they are 
perfect graphs.  Although in general characterizing stable set polytopes is \textit{NP}-hard, if the
conflict graph is claw-free, then it can be done in polynomial time
\cite{Wes:B01,KimSunMedEryKot:11}.  We use the fact that certain traffic patterns produce graphs
whose 
\vspace{-5mm}
\subsection{Summary of Main Results}
\label{sec:summary-main-results}
All results assume drives with deterministic read times.   
\begin{enumerate}
      \item  In Sec.~\ref{sec:current-insights}, we develop a \ac{QCN} model that can be used to map
  a storage network with arbitrary file layouts across drives into a queueing network, including
  systems with nonuniform file or chunk replication.  
  In the spirit of \cite{RosBam:07}, Sec.~\ref{sec:constraints} describes our \ac{QCN} model as a
  moded system, whereby valid modes are described as inequalities that capture multipacket
  reception, drives with multiple service units, and multiple unicast, broadcast, and multicast
  traffic or communication patterns.  See Fig.~\ref{fig:switchOnlyModel_Ex1} for an example.
      \item Sec.~\ref{sec:conflict-graph} describes the construction of a conflict graph from the
  \ac{QCN} model, in which we divide our analysis into systems in which drives have
  infinite or finite I/O access bandwidth.
      \item Sec.~\ref{sec:characterizing-stabg} characterizes the stable set polytope for finite I/O
  bandwidth systems.  We show that, for systems under a
multicast traffic pattern, associated conflict graphs are not guaranteed to be claw-free.  However, systems with a broadcast
  only as well as broadcast or single unicast traffic patterns result in perfect and claw-free conflict graphs.
  In a system with sufficient multipacket reception or with a multiple unicast traffic pattern, the conflict graph is a quasi-line graph.  We
  then use recent results \cite{EisOriStaVen:08} that allows the
characterization of stable set polytopes for quasi-line graphs,
which are a strict superset of perfect graphs.  See Table \ref{tab:networkScenarios} for a summary.  
      \item Sec.~\ref{sec:char-rate-regi-1} adopts and adjusts techniques from Tassiulas \etal \cite{TasEph:92} to transform our conflict graph
  characterizations into an offline scheduling algorithm that is rate optimal for uncoded storage. 
  For coded storage, Sec.~\ref{sec:effect-coded-storage} develops a \ac{RR} upper bound, for which
  we provide an intuitive interpretation, or equivalently the \ac{RR}
  given particular dynamic coding systems.
The upper bound is found by adding links into an equivalent uncoded \ac{QCN} model, which
intuitively depicts coded storage's additional initial scheduling options.  

Sec.~\ref{sec:numerical-results} then presents examples and
  numerical results showing that the \ac{RR} of coded storage can subsume that of
uncoded storage, with increases in volume averaging 50\% across traffic patterns.  
\end{enumerate}

The remainder of this paper  is organized as follows.  The general system model and basic notation is
described in Sec.~\ref{sec:graph-repr}.  Preliminaries are detailed in
Sec.~\ref{sec:preliminaries}.  The \ac{QCN} model construction is detailed in
Sec.~\ref{sec:current-insights} and the characterization of associated conflict graphs is presented
in Sec.~\ref{sec:gener-confl-graph}.  Sec.~\ref{sec:effect-coded-storage} discusses the effect of
coded storage, and Sec.~\ref{sec:numerical-results} presents examples and numerical results.
Finally, Sec.~\ref{sec:disc--concl} concludes the paper.  

\vspace{-2mm}
\section{System Model}
\label{sec:graph-repr}

We study storage systems with the following system model.  
\begin{itemize}
      \item File layout: Without loss of generality, consider a single chunked file $ \mathcal F= \{ f_{1}, \dots, f_{T}\}$ is stored in drives,
  and the $ n$th drive stores a subset of chunks $ \mathcal
  F_{n} \subseteq \mathcal F$.\footnote{$ \mathcal{F} $ may represent one or more logical
    physical files.}  (We do not consider multisets in which single drives can store multiple
  chunk replicas.)  Let the total number of chunks stored in the system be equal to
  $W = \sum_{n} \vert \mathcal{F}_{n} \vert$, and $ \mathcal{F} \subseteq \cup_{n} \mathcal{F}_{n}$.
      \item Drive behavior:  Drives have deterministic read and communication pattern of one chunk
  per timeslot per service unit.  (We do not allow preemption or processor sharing between drives.)
  Drive $ n$ has $K_{n} \in \mathbb N_{+}$ service units vis-\`{a}-vis queueing theory.  We
  refer to each service unit as a \textit{virtual drive}, and label the set of virtual drives as
  $\{D_1, \ldots, D_k,\ldots,D_{R}\}$, where $ R=\sum_{n}K_{n}$.   
   \item User management: At any given time, the system can manage up to finite $N$ active users, denoted by $ \mathcal U = \{ u_{1}, \dots,
u_{N}\}$.  These $N$ users can be, for instance, subscribers to a system or connected routers or
other aggregating nodes in a larger content distribution network.  
    \item Server behavior:  As in classic \acf{PMP} networks \cite{FouCasSerMaiMed:13}, we
consider servers that can multicast chunks read from drives to user subsets with various structures,
including a multicast traffic pattern.
\end{itemize}

Consider a queueing network composed of a set of input queues or buffers $ \mathcal{Q} ^{I}$, each of
potentially infinite size, and a set of
output lines or sinks $ \mathcal{Q} ^{O}$ connecting to outside users.  Outside users send read request for file chunks to the
network, and requests arrive at $ \mathcal{Q} ^{I}$.  When a read request is serviced, appropriate
chunks are read from one or more drives, and that
read data is then transmitted to a set of users using output lines in $ \mathcal{Q} ^{O}$.  We say
that a read request has been \textit{serviced} when that request has left its input queue, the
requested chunk has been read from drives and then completed transmission on all appropriate output lines.  

All lines have the same capacity called the \textit{line rate.}   All virtual drives have
the same deterministic capacity and read-times.  Time is slotted, where the
length of a timeslot is the reciprocal of the line rate plus the read time of a service unit.  

\vspace{-2mm}
\section{Preliminaries}
\label{sec:preliminaries}

This section will introduce select topics in queueing and graph theory used later in the paper.  It will also introduce the
reader to the communications or traffic patterns that are considered throughout the paper.  Again,
refer to Fig.~\ref{fig:method} for an illustration of the method used in this paper.  Readers
fluent in both queueing and graph theory are encouraged to immediately read
Sec.~\ref{sec:current-insights} and to use this section simply as a reference for notation.  

\begin{figure*}[tb]
  \centering
  \includegraphics[width= \linewidth]{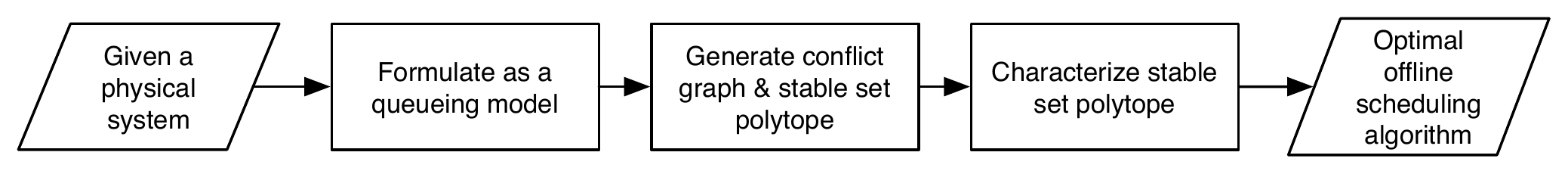}
  \caption{Summary of the technique used to generate scheduling algorithms in this paper.
    First, we formulate a physical system as a moded queueing model.  Second, we generate a conflict
    graph that depicts the constraints on the modes of this queueing model.  Third, we apply recent
    graph theory results that characterize the stable set polytope for this conflict graph.
    Fourth, we apply known techniques to translate a stable set polytope into an optimal
offline scheduling algorithm, in which case incoming traffic statistics
    are known.  The model presented in this  paper can also be applied to online scheduling systems
    in which case no knowledge of incoming traffic statistics is known, by applying policies such as
  that shown in \cite{TasEph:92}.}
  \label{fig:method}
\end{figure*}
 
\vspace{-2mm}
\subsection{Queueing Theory}
\label{sec:queu-theory-defin}

This subsection lists preliminary queueing theory definitions used throughout the paper.  The
reader is referred to \cite{Kle:B75} for a more thorough survey on queueing theory.  

\begin{definition}
  A \textit{flow} and \textit{rate} are the stream of all read request chunks, and the average
  number per timeslot, respectively, that arrive at some input   queue $ q \in \mathcal{Q} ^{I}$
  that need to  be serviced via output lines $ \mathcal{Q}
  ^{O}$.  Let $ \mathbf{r}\in \mathbb R_{+}^{\vert \mathcal{Q} ^{I}\vert}$ denote the \textit{rate
    vector} of all rates into
all input queues.  
\end{definition}

\begin{definition}
  A set of flows is called \textit{admissible} if the sum of the rates of all the flows through each input
  queue or output line does not exceed one, so inputs and outputs are not oversubscribed.  
\end{definition}

\begin{definition}
  A rate vector is said to be\textit{ achievable} if there exists a schedule that can serve it,
  while keeping all queues stable.
\end{definition}

\begin{definition}
  The \textit{rate region} is the set of all achievable rate vectors.
\end{definition}

\vspace{-2mm}
\subsection{Traffic Patterns }

We explore various storage, communication, link, and traffic patterns throughout the paper.  In
\acf{PMP} networks, a number of communication strategies or \textit{traffic patterns} are possible
between the physical drives  that read
chunks, and the users receiving those chunks, depending on
the system technologies.  Unless explicitly stated otherwise, all systems are assumed to have single packet
reception.\footnote{A packet will consist of the read chunk, as well as some overhead as required by the communication protocol. We assume that such overheads are negligible.}  

A \ac{PMP} storage system has a particular \textit{traffic pattern} if, in every timeslot, the set of feasible flows from
virtual drives to output lines must meet particular constraints.  General traffic
patterns are captured by the \acf{QCN} model  described in
Sec.~\ref{sec:current-insights}.  Specific traffic patterns analyzed in this paper are:
\begin{definition} \textit{(Single Unicast)}
A system uses a \textit{single unicast} traffic pattern if, in every timeslot, a maximum of one output
line can be used to service a chunk read by a single virtual drive.  
\end{definition}
\begin{definition}\textit{(Multiple Unicast)}
A system uses a \textit{multiple unicast} traffic pattern if, in every timeslot, each output line
can be used to service a chunk, but chunks transmitted along each output line must be read by distinct
virtual drives.  
\end{definition}
\begin{definition} \textit{(Broadcast)}
  A system uses a \textit{broadcast} traffic pattern if, in every timeslot, each output line must
  transmit the same chunk read by the same virtual drive.  
\end{definition}
\begin{definition} \textit{(Multicast)}
  A system uses a \textit{multicast} traffic pattern if, in every timeslot, each output line can be
  used to service a single chunk from any virtual drive.  
\end{definition}
\begin{definition} \textit{(Multipacket Reception)}
  A system uses \textit{$R_{x}(j)$-chunk multipacket reception} if, in each timeslot, the $j$th
  output line or channel can be used to transmit $R_{x}(j)$ unique chunks without error.  
\end{definition}
From a modeling perspective, multipacket reception (MPR) for communications can be viewed as a
generalization of speed-up in cross-bar switches
\cite{KimSunMedEryKot:11}.  Speed-up is obtained by the communication medium operating at a faster
rate, and multipacket reception can be obtained either by speed-up or by particular communication
receivers or codes being employed.

\vspace{-2mm}
\subsection{Coded Storage}
\label{sec:network-coding}

Coded storage allows physical drives to store linear combinations of chunks, as opposed to only individual
chunk subsets.  We
refer the reader to \cite{DimRamWuSuh:11} for a survey of codes for distributed storage.  We denote
the $i$th uncoded chunk as $ f_{i}$, and the corresponding $i$th coded chunk in the coded system,
stored in the same physical location, as
$ f_{i}^{c}$.  The set of chunks whose information is encoded or mixed with $ f_{i}^{c}$ is called the generation of
chunk $ f_{i}^{c}$, whose set of uncoded chunk indices we denote by $ g(i)$.  Chunks, be they coded
or not, consume the same storage space (ignoring any overhead for storing coding parameters).  In
this manuscript, we focus on $(\alpha,s)$ \ac{MDS}
codes, and do not allow the coded  chunks to update or
regenerate once they are loaded onto physical drives.  \ac{MDS} codes are those in which a set of
chunks or a generation is encoded into $ \alpha$ coded
chunks and, if a user downloads any $ s$ coded chunks, then that full generation can be decoded.   We
make two key assumptions regarding coding:
\begin{itemize}
      \item No user should receive a replica coded chunk from a generation prior to being able to
  decode that generation.
      \item If a chunk $ f_{i}$ is coded within any generation, then all replicas of $ f_{i}$ are
also coded.
\end{itemize}
A Reed-Solomon code is an example of an \ac{MDS} code.  During a particular timeslot, a user is said
to require $r$ additional \textit{degrees of freedom} if they require $r$ unique additional
coded chunks from the storage system to decode a particular generation to be able to decode chunks
in the associated generation.  Further, $ f_{i}^{wc}$ is said to be \textit{innovative} for a
user if receiving $ f_{i}^{c}$ would reduce the required degrees of freedom for that user by one.  

\vspace{-3mm}
\subsection{Graph Theory}
\label{sec:existance-claws}
This subsection lists graph theory definitions used throughout the paper.  We refer the reader to
\cite{Wes:B01} for a more thorough  survey on graph theory.  Definitions
are across graph $ G=(V,E)$, composed of vertices $V$ and edges $E$.  

\begin{definition} \textit{(Hyperedge)}
  A \textit{hyperedge} is of graph $G$ is an edge $ e\in E \subseteq \mathcal{P} (V)$, the power set
  of $V$.  In particular, a hyperedge can connect
  any number of vertices from that graph, instead of only two vertices.
\end{definition}

\begin{definition} \textit{(Incidence vector)}
The\textit{ incidence vector} of a set of vertices $ V_{1} \subseteq
  V(G)$ is a $\{0,1\}$-vector $ \mathbf x$ whose entries are labeled with the vertices of $G$.  If $
  x_{i} = 1$, then vertex $i$ is in $ V_{1}$; otherwise, $ i \notin V_{1}$.  
\end{definition}
\begin{definition} \textit{(Clique)}
  A  subgraph is called a \textit{clique} if all vertices in the subgraph are pairwise connected.  
\end{definition}

\begin{definition} \textit{(Stable set)}
  A set of vertices $ V_{i} \subseteq V$ forms a\textit{ stable set} if for every
  pair of vertices in $ V_{1}$, there is no edge connecting the two.  
\end{definition}
\begin{definition} \textit{(Stability number)}
  The\textit{ stability number} $ \alpha(G)$ is the maximum cardinality
of a stable set of $G$.
\end{definition}
\begin{definition} \textit{(Stable set polytope)}
  The\textit{ stable set polytope} $STAB(G)$ of a graph $G=(V,E)$ is
  the convex hull of the incidence vectors $\mathbf x$ of the stable sets of $G$.  
\end{definition}
\begin{definition} \textit{(Claw-free graph)}
We say that a conflict graph is
\textit{claw-free} if no induced subgraph of $G$ is a vertex with three pairwise disconnected
neighbors.
\end{definition}
\begin{definition} \textit{(Quasi-line graph)}
A graph is a \textit{quasi-line graph} if the closed
neighborhood of every vertex can be partitioned into two cliques.
\end{definition} 
\begin{definition}\textit{(Chromatic number)}
The \textit{chromatic number} of graph $ G $ is the smallest number of colors needed to color the
  vertices of $G$ so that no two adjacent vertices share the same color.
\end{definition}
\begin{definition} \textit{(Perfect graph)}
A graph is \textit{perfect} if the chromatic number of every induced
  subgraph equals the size of the largest clique of that subgraph.  
\end{definition}

\vspace{-2mm}
\subsection{Conflict Graphs}
\label{sec:confl-graph-constr}

Conflict graphs are discussed in detail in \cite{Sch:B03,GroLovSch:B93}.  A brief overview follows.
Given network graph $ G_{\text{net}}=(V _{\text{net}},E
_{\text{net}})$, and associated feasibility constraints across $ E _{\text{net}}$,
conflict graphs allow the visualization of those 
feasibility constraints.  In this paper conflict graphs are between hyperedges in the queueing
network model.  In general, conflict graph
construction generates a simple\footnote{A graph is simple if it has no loops or parallel edges.} and finite
conflict graph $ G = (V, E)$ as follows:
\begin{itemize}
      \item For every possible hyperedge $ e \in E_{\text{net}}$, create a set of vertices $ v_{(e,j)}$ in $
   V$ so that there is a one-to-one correspondence between all possible states $j$ of
hyperedge $e$ (excluding the empty set state $ \emptyset$) and the vertices $ v_{(e,j)}$.
      \item Connect vertices $ v_{(e,j)}$ and $ v_{(e',j')}$ if assigning state $j$ to $e$ and
  state $ j'$ to $ e'$ simultaneously is impossible due to a conflict across feasibility constraints \cite{KimSunMedEryKot:11}.
\end{itemize} 
A stable set from the conflict graph $STAB(G)$ represents a collection of links that can operate simultaneously
without conflict, hence it represents a valid system mode.  The \ac{SSP} can
be thought of as the convex combination of all valid modes and through timesharing, any point in the
\ac{SSP} can be set as the system operating point.  

General conflict graphs have the potential to have a large number of states and to be
computationally intractable.  Indeed, for general graphs the problem of
solving the maximum stable set problem is known to be \textit{NP}-hard.  If the 
graph has particular structure such as being claw-free, then the maximum stable set problem can be solved in polynomial
time.  However, more than 20 years after the discovery of a polynomial algorithm for the
maximum stable set problem for claw-free graphs, the explicit description or characterization of the
\ac{SSP} for claw-free graphs remains an  open problem \cite{EisOriStaVen:08}.  Recently, it was proved that if
the conflict graph is a quasi-line graph (a strict subset of claw-free graphs), then the \ac{SSP}
can be characterized exactly using the clique-family inequalities presented in \cite{EisOriStaVen:08}.
We use this result in this paper.  In addition, note that if the conflict graph is perfect
(a strict subset of quasi-line graphs),
then the \ac{SSP} can also be exactly characterized using the techniques summarized by Kim \etal
\cite{KimSunMedEryKot:11}.  

\subsection{Generating Rate Regions from Conflict Graphs}
\label{sec:STAB}
For networks composed of buffers, each of potentially infinite size, and without multicast capabilities, it is well known that the rate region $ \mathbf{R}$ is given by 
\begin{align}
 \mathbf{R} = \Bigg \{ &\rho \in \mathbb R_{0+}^{NT} \colon & \nonumber \\ & \rho \leq \sum_{m \in \mathcal{M} }
    \left. \phi_{m} \xi_{m}, \text{for some } \phi_{m} \geq 0, \sum_{m} \phi_{m} = 1  \right\} \, ,
\end{align}
where $ \xi_{m}$ are the maximum stable sets of the conflict graph, and $ \mathbb R_{0+}^{NT}$ denotes
the non-negative real $NT$-vectors \cite{RosBam:07}, where $ NT$ is the total number of input
buffers in $ \mathcal{Q}^{I}$.  This is exactly the \ac{SSP} of the traffic
pattern's conflict graph.  In addition, it has been shown in \cite{SunDebMed:07} that in infinite buffer
networks with multicast traffic patterns but no fanout-splitting, the \ac{RR} is
again the \ac{SSP} of the traffic pattern's conflict graph.  

\vspace{-2mm}
\subsection{Generating Scheduling Algorithms from Conflict Graphs}
\label{sec:devel-sched-algor}

We consider the development of offline scheduling algorithms, which require knowledge of the incoming
traffic statistics of read requests into $ \mathcal{Q} ^{I}$.  Given a cross-bar switch network, the stable set polytope from its conflict graph, and a particular
operating point that is within the stable set polytope, it has been
shown that \textit{frame-based} algorithms, with parameters appropriately  chosen, can serve any traffic
pattern in the stable set polytope and achieve 
maximum throughput \cite{KimSunMedEryKot:11}.     A frame is a set
of $F$  consecutive timeslots,  where $ F$ is the frame
size.  Frame-based schedules are specified by a sequence of $F$ mode schedules and the scheduler
cycles through these modes periodically.  The authors in
\cite{KimSunMedEryKot:11} generated conflict graphs such that each system queue can be served by a maximum of one vertex,
and although this is an assumption we shall generalize, the ideas presented herein will
rely heavily on the frame-based offline scheduling from \cite{KimSunMedEryKot:11}.

\vspace{-2mm}
\section{Queued Cross-bar Network Model}
\label{sec:current-insights}
This section presents our general queued cross-bar network (QCN) model, which is a queueing model constructed from a 
physical storage network, as per
Sec.~\ref{sec:graph-repr}.  We generate this queueing network as follows.  
\begin{itemize}
      \item For chunk $ f_{i}$ and user $ u_{j}$, there exists one infinite size input queue labeled $q_{f_{i},
u_{j}}$.  Consider
$ \mathcal Q^{I}$ the set of $ T\times N$ input queues.
    \item  For every user $ u_{j}$ we create
an output sink $ q_{u_{j}}$, and denote $ \mathcal Q^{O}$ the set of $N$
output sinks.  
    \item As a reminder of the system model, for each service unit on each physical drive, we create a \textit{virtual drive} $ D_k$.  In the
context of the \ac{QCN} model, when we refer to drives we are always referring to virtual drives.
    \item For any $ f_{i}$ that drive $ D_k$ can read, we draw an edge from every queue in $ \{ q_{f_{i}u_{j}} \in \mathcal Q^{I} \colon j
\in \{ 1, \ldots, N \}\}$, i.e., corresponding to chunk $
f_{i}$ for any user, to output line $ q_{u_{j}} \in \mathcal Q^{O}$ with label $ D_k$
(labels are not necessarily unique).
\end{itemize}

System connectivity is represented using three matrices.  Let chunk-drive connectivity be defined by $ \mathcal N$, a $ T\times R $ matrix such that each
element is defined as
\begin{align}
  \mathcal N (i,k) = 
  \begin{cases}
    1 \quad & \text{if } f_{i} \in \mathcal F_{k} \\
0 & \text{otherwise.}
  \end{cases}
\end{align}
During each timeslot, set the status of each user's file knowledge through a $T\times N\times R $ matrix $ \mathcal{S} $,
such that for a particular timeslot, each element is defined as
\begin{align}
  \mathcal{S} (i,j,k) = 
  \begin{cases}
    1 \quad & \text{if chunk } f_{i} \text{, that is stored on drive } D_{k} \text{, would be}\nonumber \\
& \text{innovative for user } u_{j}  \text{ in this timeslot} \\
0 & \text{otherwise.}
  \end{cases}
\end{align}
 The \textit{mode set} is given as follows.  Let $ \mathcal M = \{ \mathcal M_{m} \}$ be
the set of all modes,
where $ \mathcal M_{m}$ is the $ m$th mode which is an 
$ T \times N\times R$ matrix where each element is defined as  
\begin{align}
  \mathcal M_{m} (i,j,k) = 
  \begin{cases}
    1 \quad & \text{if user $u_{j}$ receives chunk $f_{i}$} \\ & \text{via drive $D_{k}$} \\
0 & \text{otherwise.}
  \end{cases}
\end{align}
Let chunk-drive usage indicator $ r_{m}(i,k)$ be defined as,
\begin{align}
  r_{m}(i,k) =
  \begin{cases}
    1 \quad & \sum_{j} \mathcal M_{m}(i,j,k) \geq 1\\
    0 & \text{otherwise.}
  \end{cases}
\end{align}

We call the system the \textit{\acf{QCN} model.}  See
Fig.~\ref{fig:switchOnlyModel_Ex1} for a simple illustration.  

\begin{figure}[h]
  \centering
  \includegraphics[width=\linewidth]{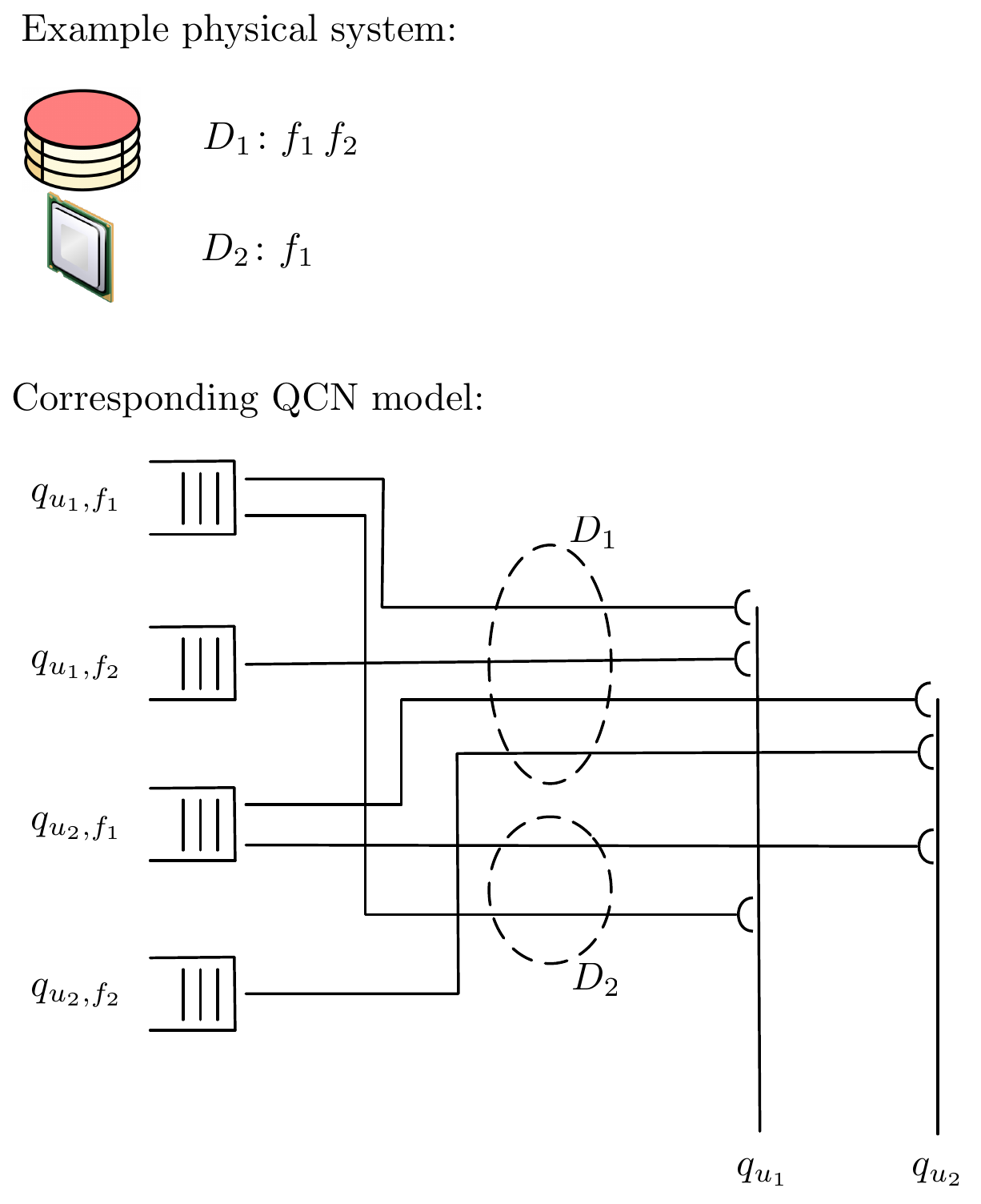}
  \caption{Example illustration of the \acf{QCN} method, modeling a physical system with two
    users, and two drives storing overlapping chunk
    sets.  }
  \label{fig:switchOnlyModel_Ex1}
\end{figure}

\subsection{Constraints}
\label{sec:constraints}

The set of constraints required for mode $ \mathcal M_{m}$ to be \textit{valid} are as follows.  
\begin{itemize}
      \item A user never receives a non-innovative chunk
  \begin{align}
    \mathcal{M} _{m}(i,j,k) \leq \mathcal{S} (i,j,k) \, , \quad \forall i,j,k
  \end{align}
      \item Up to $R_{x}(j)$ chunks can be received by the $j$th user in each timeslot, i.e., user $j$
  has $R_{x}$-chunk multipacket reception capability
  \begin{align}
\label{eq:MPRConstraint}
    \sum_{i,k} \mathcal M_{m}(i,j,k) \, \mathcal N(i,k) \leq R_{x}(j)\, , \quad \forall j \in \{ 1,\ldots , N \}
  \end{align}
      \item The $k$th (virtual) drive allows up to one read per timeslot
  \begin{align}
  \label{eq:QSSConstraint_kthDriveAllowsOneReadPerTimeslot}
    \sum_{i} r_{m}(i,k) \leq 1 \, , \quad \forall k \in \{ 1, \ldots, R\} \, .
  \end{align}
    \item Traffic pattern constraints:
\begin{itemize}
\item  Single unicast constraint: Only a single chunk can be transmitted to a single output line
\begin{align}
\sum_{i,j,k} \mathcal{M} _{m}(i,j,k) \mathcal{N} (i,k) \leq 1
\end{align}
      \item Multiple unicast constraint: Only up to one user can
receive a chunk $ f_{i}$ from the same drive $ D_k$
\begin{align}
      \sum_{j} \mathcal M_{m}(i,j,k) &\, \mathcal N(i,k) \leq 1 \nonumber \\ &\forall (i,k) \in \{ 1, \ldots , T \} \times
  \{1,\ldots , R \} \, 
\end{align}
    \item Broadcast constraint:  If a read chunk $ f_{i}$ is transmitted from $ D_{k}$ to a user,
then it is transmitted to all users
\begin{align}
&  \exists \mathcal{Y}\subseteq\{ (i,k)\colon \mathcal{N}(i,k) = 1\} \nonumber \\
&s.t. \nonumber \\
& \sum_{j} \mathcal M_{m}(i,j,k) \, \mathcal N(i,k) = N \quad \forall (i,k) \in \mathcal{Y}\\
& \sum_{j} \mathcal M_{m}(i',j,k') \, \mathcal N(i',k') = 0 \quad \forall (i',k') \notin \mathcal{Y} \nonumber
\end{align}
    \item Multicast constraint: Up to $N$ users can receive the same chunk from the same drive
\begin{align}
  \sum_{j} \mathcal M_{m}(i,j,k) &\, \mathcal N(i,k) \leq N \nonumber \\ & \forall (i,k) \in \{ 1, \ldots , T \} \times
  \{1,\ldots , R \} \, .
\end{align}
\end{itemize}

\end{itemize}
The \ac{QCN} model incorporates arbitrary
chunk layouts, service unit numbers, as well as traffic patterns, which may make it a useful tool in
high-traffic storage network analysis.  Example
systems with the traffic patterns described above are as follows.  Some systems may be restricted to multiple unicast traffic patterns if they do not
allow caching, others may be restricted to broadcast traffic patterns if they are transmitting
wirelessly.  Systems can be modeled as having  multipacket reception if their scheduling is done at
the level of frames, or multiple timeslots if output buffers exist on each
output line.

\subsection{Properties}
\label{sec:correctness}
To explore the storage system types that the \ac{QCN} model can capture and model, we introduce the following properties.  

\begin{definition} \textit{(Conservation of flow:)}
We say that a queueing model has the \textit{conservation of flow } if in any timeslot, the sum of
the number
of read requests serviced across input queues $ \mathcal Q^{I}$ is equal to the sum of the number
of chunks received at output lines $ \mathcal Q^{O}$, regardless of traffic pattern
constraints. 
\end{definition}

\begin{definition} \textit{(Multiple service unit property:)}
We say that a queueing model has the \textit{multiple service unit property}, if for any physical
drive $n$, up to any fixed $K_{n} \in \mathbb
N_{+}$ unique and stored chunks can be read from physical drive $n$ in any timeslot.
\end{definition}
A queueing model that has both conservation of flow and the multiple service unit property can
be used to model a large variety of systems, regardless of chunk-to-drive layouts.  
\begin{thm}
Any \ac{QCN} model has conservation of flow.
\end{thm}
\begin{proof}
Any queueing network with fixed topology and in which all service units have deterministic service
time has conservation of flow by construction.  For a \ac{QCN} model, in each timeslot a single mode
is selected and so we need to check that no mode exists which does not have conservation of flow.

Suppose mode $\mathcal{M} _{m}$ does not have conservation of flow.  In this case, either we service more read requests from $
\mathcal{Q} ^{I}$ than are received at $ \mathcal{Q} ^{O}$, or we transmit more chunks down output
lines than are serviced across $ \mathcal{Q} ^{I}$.  

Suppose mode $ \mathcal{M} _{m}$ services more read requests than are received by users.
This implies there exists output line $ q_{u_{j}}$ that receives less chunks than are serviced at
queues $ \{ q_{f_{i},u_{j}} \}_{i=1}^{T} $.  $ \mathcal{N} (i,k)$  implies any valid mode can only
activate edges on virtual drives with chunks in demand, and owing to
(\ref{eq:QSSConstraint_kthDriveAllowsOneReadPerTimeslot}), there is either overflow servicing from the
same chunk on different
virtual drives, or from different chunks from different drives.  Suppose it is the same chunk on
different virtual drives; if those cross bars are activated then $ \sum_{k} \mathcal{M} _{m}(i,j,k)$
cross bars are activated.  The only way for such invalid overflow to occur is if $ \sum_{k}
\mathcal{M} _{m}(i,j,k) > R_{x}(j)$, and by (\ref{eq:MPRConstraint}) we have a contradiction.  The
same contradiction holds for different chunks on different virtual drives.  

Suppose mode $ \mathcal{M} _{m}$ services fewer read requests than are received by users.  This
implies there exists a queue $ q_{f_{i},u_{j}}$ that transmits a request to more than $ q_{u_{j}}$.
However, by construction no such labeled links exist in the \ac{QCN} model, so we have a contradiction.

\end{proof}

\begin{thm}
  The \ac{QCN} model has the multiple service unit property.
\end{thm}
\begin{proof}
This follows naturally from the construction of the \ac{QCN} queueing model from a physical network.  If a physical
drive $n$ has $ K_{n}$ service units, we construct $ K_{n}$ virtual drives, all with access to physical
drive $n$'s chunks.  All constraints are indexed across virtual drives in $k$, so the property
holds by construction.  
\end{proof}

\section{Rate Region Characterization}
\label{sec:gener-confl-graph}

In this section we generate our conflict graph and then characterize the rate region.  

\vspace{-2mm}
\subsection{Conflict Graph}
\label{sec:conflict-graph}

Conflict graphs are constructed as follows.  We use hyperedges as
defined by traffic pattern constraints, so that all data is transmitted along the same fingers of
the same hyperedge.  

In our conflict graph analysis, we restrict ourselves to edge-based conflict graphs.  A conflict graph
can be  \textit{edge-based} if $ R_{x}(j) \in \{ 1, T, T+1, T+2, \ldots
\} \,\, \forall j$; for $R_{x}(j)$-multipacket reception with $ 1<R_{x}(j) <T$, then
  the conflict graph may require hyperedges to capture the selection of different copies of
  chunks.  

Drive I/O access bandwidth---tightly coupled to the number of service units per drive---is a key
parameter of modern storage systems.  As such, to build  up
intuition, we begin by analyzing simpler systems with infinite I/O access bandwidth, and then move
to finite bandwidth systems.  An infinite I/O system model would be helpful when individual drive
blocking is \textit{not} a bottleneck, and instead traffic patterns are the main constraint so  wish
to analyze their effects  directly.

\subsubsection{Infinite I/O Access Bandwidth} Consider the case when $ K_{n} \geq T \, , \,\,\, \forall n$.
In this scenario, a drive can read out all chunks simultaneously  in a
  timeslot, so no intra-drive conflicts can arise.    

As defined prior, $ \mathcal{U} = \{ u_{1}, \ldots,u_{N}\} $ is the set of active users in the
system.  To simplify the problem and reduce vertex numbers, we define valid conflict graph vertices using the
set of valid traffic pattern constraints.  The connectivity between these vertices is then set by storage and link constraints.

\textit{Vertices: } For every chunk stored in the system $ f_{i}$, we ignore the drive that
stores it since each drive has infinite I/O bandwidth.  Given a particular conflict graph
traffic pattern constraint, we generate a set of vertices in our conflict graph as 
\begin{itemize}
      \item Multicast: $\{ v_{f_{i},u_{\mathcal{S} }} \}_{\mathcal{S} \in \mathcal{P}_{\geq
      1}(\mathcal{U})  }$, where $\mathcal{P}_{\geq 1} (\mathcal{U})$ is the powerset of all
  active users excluding the empty set of users.  The total number of vertices in the conflict graph then scales as $
  \mathcal{O}(T2^{N})$.
      \item Broadcast: $\{v_{f_{i},u_{\mathcal{S}}}\}_{\mathcal{S} \in \mathcal{P}_{\geq N}
    (\mathcal{U})}$, i.e., a single $ \{ v_{f_{i},u_{\forall}}\}$ vertex.  The total number of
  vertices in the conflict graph scales as $ \mathcal{O} (T)$.
      \item Multiple unicast: $ \{ v_{f_{i},u_{j}}\}_{j=1}^{N}$, i.e., a set of $N$ vertices.  The total
  number of vertices in the conflict graph scales as $ \mathcal{O} (TN)$.  
\end{itemize}

\textit{Edges:}  Consider two vertices generated by the traffic pattern constraints, $
v_{f_{i_{1}},u_{\mathcal{S}_{1}}}$, and $ v_{f_{i_{2}},u_{\mathcal{S}_{2}}}$.  Given some $ k_{1}$ and
$ k_{2}$ such that $ \mathcal{N} (i_{1},k_{1})= \mathcal{N} (i_{2},k_{2}) = 1$, if setting
\begin{align}
 \mathcal{M}_{m}(i_{1}, j_{1},k_{1})=1 \quad \forall j_{1} \in \mathcal{S} _{1}
\end{align}
and
\begin{align}
  \mathcal{M}_{m}(i_{2}, j_{2},k_{2})=1 \quad \forall j_{2} \in \mathcal{S} _{2}
\end{align}
violates at least one storage or link constraint as in Sec.~\ref{sec:constraints}, then connect $ v_{f_{i_{1},u_{\mathcal{S}_{1}}}}$,
and $ v_{f_{i_{2},u_{\mathcal{S}_{2}}}}$ with an edge.  

\vspace{2mm}
\subsubsection{Finite I/O Access Bandwidth} 

Consider the case when $ K_{n} < T \, , \,
  \, \forall n$.
 
\textit{Vertices: }  Similarly to the infinite I/O scenario, we generate viable vertices via our
traffic pattern constraints.  The primary addition is that we generate separate vertices for
duplicate chunks stored on other virtual drives.  

Given our chunk-drive connectivity matrix $ \mathcal{N}$, for every non-zero
element of $\mathcal{N}(i,k)$ we generate a set of vertices in our conflict graph, given by our 
traffic pattern constraints as:
\begin{itemize}
      \item Multicast: $\{ v_{f_{i,k},u_{\mathcal{S} }} \}_{\mathcal{S} \in \mathcal{P}_{\geq
      1}(\mathcal{U})  }$, where $\mathcal{P}_{\geq 1} (\mathcal{U})$ is the powerset of all
  active users excluding the empty set of users.  The total number of vertices then scales as $
  \mathcal{O}(TR2^{N})$.  As a reminder, $T$ is the number of chunks, $R$ the number of virtual
  drives, and $N$ the number of users.  
      \item Broadcast: $\{v_{f_{i,k},u_{\mathcal{S}}}\}_{\mathcal{S} \in \mathcal{P}_{\geq N}
    (\mathcal{U})}$, i.e., a single $ \{ v_{f_{i,k},u_{\forall}}\}$ vertex.  The total number of
  vertices scales as $ \mathcal{O} (TR)$.
      \item Multiple unicast: $ \{ v_{f_{i,k},u_{j}}\}_{j=1}^{N}$, i.e., a set of $N$ vertices.  The total
  number of vertices scales as $ \mathcal{O} (TRN)$.  
\end{itemize}

\textit{Edges:}  Consider two vertices generated by the traffic pattern constraints, $
v_{f_{i_{1},k_{1}},u_{\mathcal{S}_{1}}}$, and $ v_{f_{i_{2},k_{2}},u_{\mathcal{S}_{2}}}$.  If setting
\begin{align}
 \mathcal{M}_{m}(i_{1}, j_{1},k_{1})=1 \quad \forall j_{1} \in \mathcal{S} _{1}
\end{align}
and
\begin{align}
  \mathcal{M}_{m}(i_{2}, j_{2},k_{2})=1 \quad \forall j_{2} \in \mathcal{S} _{2}
\end{align}
violates at least one storage or link constraint from Sec.~\ref{sec:constraints}, then connect $ v_{f_{i_{1},k_{1}},u_{\mathcal{S}_{1}}}$,
and $ v_{f_{i_{2},k_{2}},u_{\mathcal{S}_{2}}}$ with an edge.  

\vspace{2mm}
As a simple example of a conflict graph, consider a storage system with two users $u_1$ and $u_2$
and two chunks $f_1$ and $f_2$ stored on two drives as $D_1: f_1$ and $D_2:f_2$.  The conflict graph
of this system under multicast traffic pattern with $R_x(1)=R_x(2)=1$ and drives having single
service units is shown in Fig.~\ref{fig:example_conflict_graph}(a). In this graph, as can be seen,
since each user can receive only up to one chunk at each timeslot, the vertices that correspond to
the same user conflict with each other. In addition, since drives have single service units, the
vertices that correspond to the same drive conflict with each other.  Furthermore, the conflict graph of the above system under multicast traffic pattern with $R_x(1)=R_x(2)=2$ is shown in Fig.~\ref{fig:example_conflict_graph}(b). In this graph, since multipacket reception is allowed, receivers are able to receive up to two chunks per timeslot. Therefore, there is no conflict between the vertices of the same user. However, due to drives' single service units, still the vertices that correspond to the same drive conflict with each other.
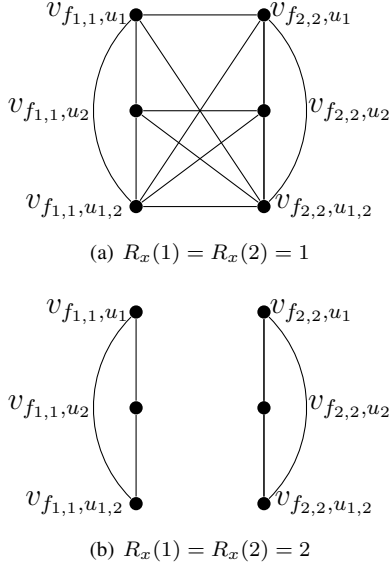
\begin{figure}[thb!]
\centering
\subfigure[$R_{x}(1) = R_{x}(2) = 1$]{\begin{tikzpicture}[scale=0.85]
\node[vertex,minimum size=5pt,color=black,text=black] (v1) at (-1,1.5)  {};
\node[vertex,minimum size=5pt,color=black,text=black] (v2) at (-1,0)  {};
\node[vertex,minimum size=5pt,color=black,text=black] (v3) at (1,1.5)  {};
\node[vertex,minimum size=5pt,color=black,text=black] (v4) at (1,0)  {};
\node[vertex,minimum size=5pt,color=black,text=black] (v5) at (-1,-1.5) {};
\node[vertex,minimum size=5pt,color=black,text=black] (v6) at (1,-1.5) {};
\node [font=\large] at (-1.75,1.5) {$v_{f_{1,1},u_1}$};
\node [font=\large] at (1.75,1.5) {$v_{f_{2,2},u_1}$};
\node [font=\large] at (-2.35,0) {$v_{f_{1,1},u_2}$};
\node [font=\large] at (2.35,0) {$v_{f_{2,2},u_2}$};
\node [font=\large] at (-1.95,-1.5) {$v_{f_{1,1},u_{1,2}}$};
\node [font=\large] at (1.95,-1.5) {$v_{f_{2,2},u_{1,2}}$};
\draw[-] (v1) to (v2);
\draw[-] (v1) to (v3);
\draw[-] (v1) to  [out=-135,in=135] (v5);
\draw[-] (v1) to (v6);
\draw[-] (v2) to (v4);
\draw[-] (v2) to (v5);
\draw[-] (v2) to (v6);
\draw[-] (v3) to  [out=-45,in=45] (v6);
\draw[-] (v3) to (v4);
\draw[-] (v3) to (v5);
\draw[-] (v3) to (v6);
\draw[-] (v4) to (v5);
\draw[-] (v4) to (v6);
\draw[-] (v5) to (v6);
\end{tikzpicture}}
\subfigure[$R_x(1)=R_x(2)=2$]{\begin{tikzpicture}[scale=0.85]
\node[vertex,minimum size=5pt,color=black,text=black] (v1) at (-1,1.5)  {};
\node[vertex,minimum size=5pt,color=black,text=black] (v2) at (-1,0)  {};
\node[vertex,minimum size=5pt,color=black,text=black] (v3) at (1,1.5)  {};
\node[vertex,minimum size=5pt,color=black,text=black] (v4) at (1,0)  {};
\node[vertex,minimum size=5pt,color=black,text=black] (v5) at (-1,-1.5) {};
\node[vertex,minimum size=5pt,color=black,text=black] (v6) at (1,-1.5) {};
\node [font=\large] at (-1.75,1.5) {$v_{f_{1,1},u_1}$};
\node [font=\large] at (1.75,1.5) {$v_{f_{2,2},u_1}$};
\node [font=\large] at (-2.35,0) {$v_{f_{1,1},u_2}$};
\node [font=\large] at (2.35,0) {$v_{f_{2,2},u_2}$};
\node [font=\large] at (-1.95,-1.5) {$v_{f_{1,1},u_{1,2}}$};
\node [font=\large] at (1.95,-1.5) {$v_{f_{2,2},u_{1,2}}$};
\draw[-] (v1) to (v2);
\draw[-] (v1) to  [out=-135,in=135] (v5);
\draw[-] (v2) to (v5);
\draw[-] (v3) to  [out=-45,in=45] (v6);
\draw[-] (v3) to (v4);
\draw[-] (v3) to (v6);
\draw[-] (v4) to (v6);
\end{tikzpicture}}
\caption{ Conflict graphs of a physical network with two users $u_1$ and $u_2$ and two chunks
  $f_1$ and $f_2$, operating with a multicast traffic pattern.}
\label{fig:example_conflict_graph}
\end{figure}

\vspace{-2mm}
\subsection{Characterizing the SSP}
\label{sec:characterizing-stabg}
As a reminder, we restrict our characterization to edge-based conflict graphs.  We provide
characterization analysis for the conflict graphs generated in the prior subsection for finite I/O
access bandwidth systems; characterizations for infinite I/O systems are extremely similar.  

 Unless otherwise stated, we do not allow multipacket reception so $ R_{x}(j) = 1, \,
\, \forall j$.  In multipacket reception systems we use the multicast traffic pattern.  The
framework we have setup thus far provides rapid analysis of many storage systems with various
traffic patterns.  See Table
\ref{tab:networkScenarios} for scenarios for which we have characterized the associated conflict
graphs, and by extension their associated \acp{SSP}.

\begin{table}[tb]
  \centering
  \caption{\ac{QCN} models with finite I/O: Traffic patterns and their associated conflict graph
    properties.  }
\vspace{-5mm}
  \label{tab:networkScenarios}
\bigskip
\ra{1.3}
  \begin{tabular}{@{}l c l @{}}
    \toprule
    \textbf{Traffic Pattern} & \textbf{Claw-free} & \textbf{Notes} \\
   \midrule
Single unicast & Yes& Perfect (Lem.~\ref{lem:INFIO_SU_B_BSU})\\
Broadcast & Yes& Perfect (Lem.~\ref{lem:INFIO_SU_B_BSU}) \\
   Broadcast and single unicast &Yes& Perfect (Lem.~\ref{lem:INFIO_SU_B_BSU})\\
Multiple unicast & Yes& Quasi-line (Lem.~\ref{lem:INFIO_MU})\\
Multicast&No& (Lem.~\ref{lem:multicast_notClawFree})\\
   Broadcast and multiple unicast&No& (Cor.~\ref{cor:Broadcast_MU_notClawFree})\\
$\geq T$-multipacket reception &Yes&Quasi-line (Thm.~\ref{thm:INFIO_MPR_QuasiLine})\\
$<T$-chunk multipacket reception && Unknown\\
   \bottomrule
  \end{tabular}
\end{table}

Given a multicast traffic pattern, the general conflict graph is not claw-free, making exact
\ac{SSP} characterization challenging and motivating the exploration of more restrictive traffic patterns.  
\begin{lem}
\label{lem:multicast_notClawFree}
Given a \ac{QCN} model operating with a multicast traffic pattern, with
greater than two users and greater than two chunks, then
  the associated conflict graph is not guaranteed to be claw-free.
\end{lem}

\begin{proof}

  Suppose the conflict graph is claw-free.  Consider the counter example given by
  Fig.~\ref{fig:counterExampleToClaw}.   Given there are $ T\geq 3$ chunks in our target file and
  $ N\geq 3$ users, consider the broadcast traffic pattern.  Owing to the lack of multipacket reception,
  the set of vertices $ \{v_{f_{i,k},u_{\forall}} \}_{i=1}^{T}$ form a single clique, where we use the simplifying notation $ u_{\forall} =
  \{u_{1},\ldots, u_{N}\}$.  For a given chunk $ f_{i,k}$, the set of vertices $
  \{ v_{f_{i,k},u_{j}} \}_{i=1}^{N}$ form a clique, as do vertices $ \{ v_{f_{i,k},u_{j}} \}_{i=1}^{T}$ for a given user
  $u_{j}$.  In Fig.~\ref{fig:counterExampleToClaw} we depict each of these cliques as a rectangle.
  As highlighted in red with curved edges, the set of unique vertices $ \{
  v_{f_{i,1},u_{\forall}},v_{f_{i,1},u_{j}},v_{f_{i_{2,1}},u_{j_{2}}},v_{f_{i_{3,1}},u_{j_{3}}}  \}$ form a
  claw.  

\begin{figure}[tb]
  \centering
  \includegraphics[width=\linewidth]{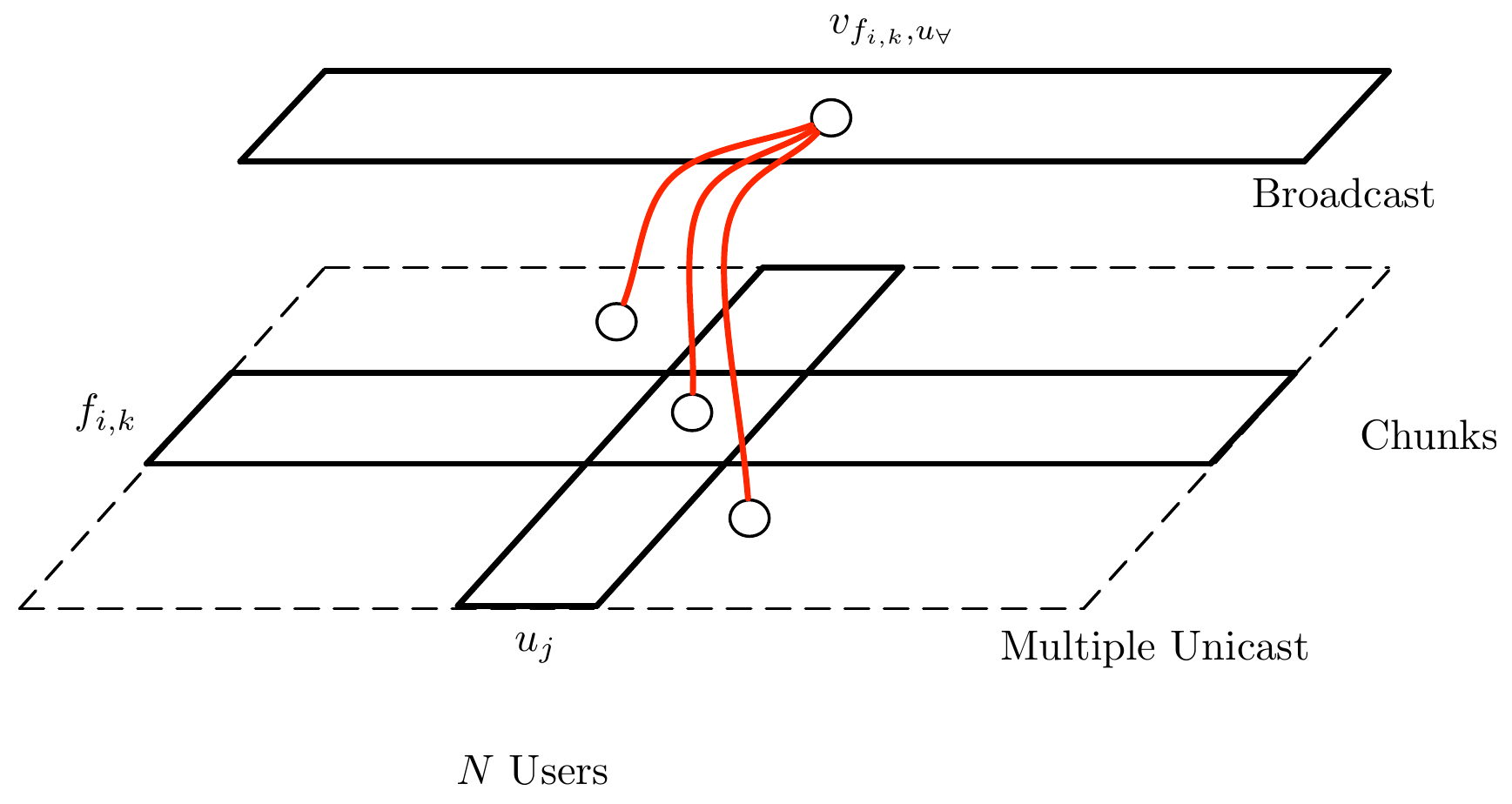}
  \caption{As per the proof of Lemma \ref{lem:multicast_notClawFree}, given a multicast traffic
    pattern, an example of a conflict graph with a claw.  A clique of vertices is depicted as a
    rectangle.  The graph depicts a subset of the overall conflict graph, illustrating broadcast and
    multiple unicast clique structure.}
  \label{fig:counterExampleToClaw}
\end{figure}

\end{proof}
The presented counterexample immediately generates an additional corollary: Systems
that allow both multiple unicast and broadcast traffic patterns are not guaranteed to admit claw-free conflict graphs.
\begin{cor}
\label{cor:Broadcast_MU_notClawFree}
Given a \ac{QCN} model operating with a traffic pattern that allows either
broadcast or multiple unicast, with greater than two users and greater than two chunks,  then the associated
conflict graph is not guaranteed to be claw-free.
\end{cor}
\begin{proof}
As per the counterexample shown in Fig.~\ref{fig:counterExampleToClaw}.
\end{proof}
We now consider systems with restricted traffic patterns, including single unicast, broadcast, and broadcast with
single unicast.  
\begin{lem}
\label{lem:INFIO_SU_B_BSU}
Suppose a \ac{QCN} model operates  with a traffic pattern of either single unicast,
broadcast, or broadcast and single unicast.  Then its associated conflict graph is
  claw-free, perfect, and with stability number equal to one.
\end{lem}
\begin{proof}
First, consider a single unicast traffic pattern.  The set of associated conflict graph vertices $ \{
v_{f_{i,k},u_{j}} \}_{i,j,k}$ form a single clique.  Second, consider a broadcast traffic pattern.
The set of associated conflict graph vertices $ \{ v_{f_{i,k},u_{\forall}}\}_{i,k} $ again form a
single clique.  Third, consider the broadcast or single unicast traffic pattern.  The
associated conflict graph has one clique from the broadcast and another from the single unicast traffic
pattern.  Owing to there being no multipacket reception, these two cliques are both fully connected.  Cliques
have stability number of one and are an example of perfect graphs.  
\end{proof}
\begin{lem}
  \label{lem:INFIO_MU}
Suppose a \ac{QCN} model operates with a multiple unicast traffic pattern.  Then
its associated conflict graph is a quasi-line graph.
\end{lem}
\begin{proof}
In constructing the conflict graph, there exist two constraints of note.  First, each virtual drive
can make up to one transmission, as per (\ref{eq:QSSConstraint_kthDriveAllowsOneReadPerTimeslot}).
Second, each user can receive up to one unicast chunk, as per the multiple unicast traffic pattern.
Consider conflict graph vertex $ v_{f_{i,k},u_{S}}$.  It is a member of two cliques: first, across all
chunks on drive $D_{k}$ and across all $N$ users.  We illustrate this as the horizontal clique
panels in Fig.~\ref{fig:multiunicast_multiChunksPerDrive}.  Second, given user set $ u_{S}$, all
drives and chunks form another clique.  We illustrate this as the vertical clique panels in
Fig.~\ref{fig:multiunicast_multiChunksPerDrive}.  Not all cliques in
Fig.~\ref{fig:multiunicast_multiChunksPerDrive} need be of the
same size.  There are no other connections or conflicts between vertices.  
The conflict graph is then a quasi-line graph.  
\end{proof}

\begin{figure}[tb]
  \centering
  \includegraphics[width=0.9\linewidth]{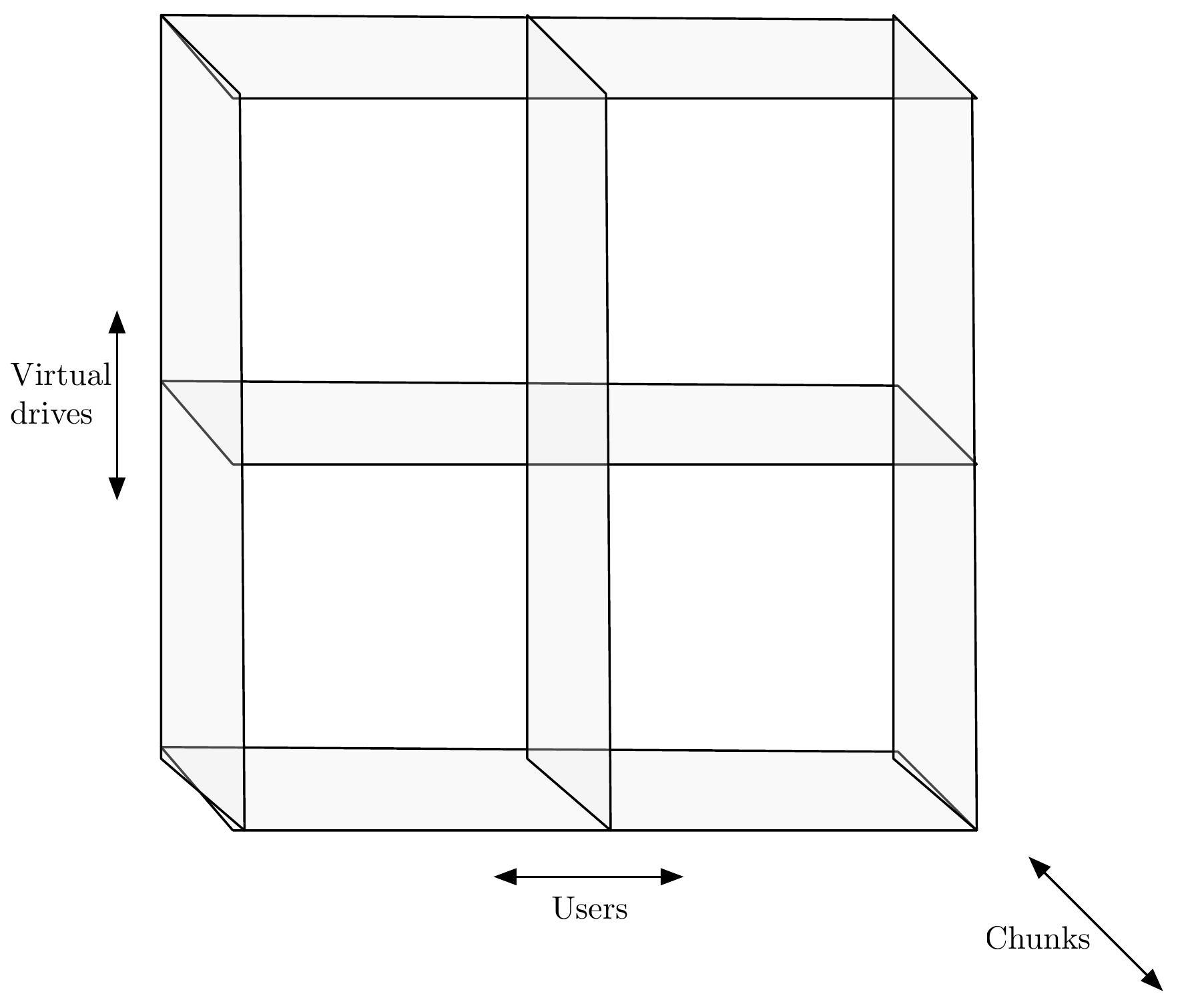}
  \caption{Simplified visualization of a conflict graph associated with a finite I/O \ac{QCN} model
    operating with a multiple unicast traffic pattern.  Each horizontal panel represents a clique
generated by virtual drive transmissions, and each vertical panel represents a clique generated by
user traffic pattern restrictions.}
  \label{fig:multiunicast_multiChunksPerDrive}
\end{figure}

\begin{thm}
\label{thm:INFIO_MPR_QuasiLine}
  Suppose a \ac{QCN} model operates with any traffic pattern, and all users have
  at least $T$-chunk multipacket reception.  Then its associated conflict graph is a quasi-line graph.  
\end{thm}
\begin{proof}
Consider chunk $ f_{i,k}$.  The vertices generated in the conflict graph associated with $ f_{i,k}$ are
a function of the traffic pattern used and the drive service unit constraints.  Regardless of traffic pattern particulars, all
vertices associated with $ f_{i,k}$ will
form a clique owing to drive constraints, i.e., for any $ v_{f_{i,k},u_{S_{1}}}$ and $ v_{f_{i,k},u_{S_{2}}}$, they must be connected, where $
S_{1},S_{2} \subseteq \{1, \ldots, N \}$.  Crucially, however, owing to multipacket reception, unique
vertices $v_{f_{i,k},u_{S_{1}}}$ and $ v_{f_{j,l},u_{S_{2}}}$, $ i \neq j, \, k \neq l$ are not connected.  The conflict
graph is then a set of disjoint cliques, where each clique is associated with a single chunk on a
single drive.  This set of disjoint cliques is a quasi-line graph as each closed neighborhood can be partitioned into the union of two cliques.  
\end{proof}

We also point out a connection between claw- and \textit{net-free} conflict graphs and \ac{QCN} models.   In
graph theory, the study of claw-free graphs is often associated with claw- and net-graphs, which are both claw-  and net-free.  A net is illustrated in Fig.~\ref{fig:netEx}, which is
formed by starting with a triangle and adding to each vertex a new vertex.  More is known
about claw- and net-free graphs than only claw-free graphs.  In the generation of our conflict graphs we explicitly do not consider null or do-not-transmit
vertices, as per the our hyperedges no including the emptyset.  We point out that if do-not-transmit vertices are included in conflict graph generation,
then those conflict graphs are not guaranteed to be net-free.
\begin{figure}[tb]
  \centering
  \includegraphics[width=0.32\linewidth]{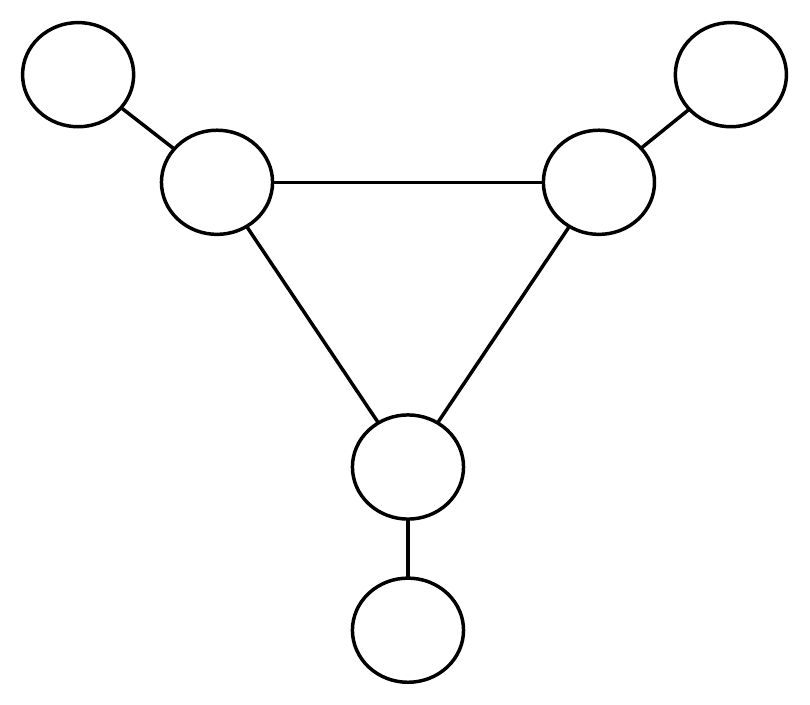}
  \caption{Net graph illustration.  It is formed by starting with a triangle and adding a new
vertex to each original vertex.}
  \label{fig:netEx}
\end{figure}
\begin{lem}
\label{lem:NetFreeOrNot}
If do-not-transmit vertices are used in the construction of conflict graph $ G$, then $ G$ is not guaranteed to be 
net-free.  
\end{lem}
\begin{proof}
  Consider a \ac{QCN} model with infinite I/O operating with a broadcast traffic pattern, with three
  chunks $ T=3$ on a virtual drive, and one user $ N=1$.
  Using do-not-transmit vertices for the conflict graph generation, we have three transmit vertices $ \{v_{f_{i,1},u_{1}} \}_{i=1}^{3}$ which
  form a central clique.  We then also generate three additional vertices $ \{
  v_{f_{i,1},u_{\emptyset}} \}_{i=1}^{3}$, where each vertex in this set denotes not transmitting
  chunk $ f_{i,1}$.  We then connect vertex pairs $ (v_{f_{i,1},u_{1}}, v_{f_{i,1},u_{\emptyset}})$
  with edges for each chunk, and we have a net graph.  
\end{proof}

Given a conflict graph whose \ac{SSP} can be characterized, as per Table
\ref{tab:networkScenarios}, we now describe how to
characterize the rate regions based on these \acp{SSP}.  

\vspace{-2mm}
\subsection{Characterizing the Rate Region}
\label{sec:char-rate-regi-1}

This subsection establishes the achievable rate region of \ac{QCN} systems in terms of their associated conflict graphs,
where the \ac{SSP} has been characterized.  We find the rate regions 
by adapting scheduling policy $ \pi_{0}$ from \cite{TasEph:92}---which was also adapted to become
offline scheduling Algorithm 1 from \cite{KimSunMedEryKot:11}---to operate on
\ac{QCN} models.  The
main difference between this paper and policy $ \pi_{0}$ from \cite{TasEph:92} is that our approach
also allows traffic pattern selection.  Further, in Algorithm 1 from
\cite{KimSunMedEryKot:11} each queue in $ \mathcal{Q}^{I} $ is served by a maximum of one vertex in
the conflict graph.  In this work, multiple vertices can service any particular queue.  As such, we
introduce two types of incidence vectors:
\begin{definition} \textit{($(f_{i,k},u_j)-$ incidence vector):}
An $(f_{i,k},u_j)-$incidence vector $\mathcal C^{i,k,j}$ is a $\{0,1\}-$vector of size $m$, $m$
being the total number of stable sets in the conflict graph, whose entries are  labeled with the
stable sets $S^\ell, \forall \ell\in \{1,\cdots,m\}$.  If $\chi^{S_\ell}({v_{f_{i,k},u_{\mathcal
      S\in \mathcal P_{\geq 1}(\mathcal U_A)}}})=1$ where  $j\in \mathcal S$, then $\mathcal C^{i,k,j}(\ell)=1$; otherwise $\mathcal C^{i,k,j}(\ell)=0$.
\end{definition}
\begin{definition} \textit{($(f_i,u_j)-$incidence vector):}
To obtain the $(f_i,u_j)-$incidence vector $\mathcal C^{i,j}$, we add up all $(f_{i,k},u_j)-$
incidence vectors $\forall k\in \{1,...,R\}$,
\begin{eqnarray}
\mathcal C^{i,j}=\sum_{k=1}^{R}{\mathcal C^{i,k,j}} \, .
\end{eqnarray}
\end{definition}

Recall that frame-based schedules can be specified by a sequence of $F$ switch configurations
such that the switch cycles through these configurations 
periodically.  These schedules are decided based on prior knowledge of the arrival rates of the
flows, and do not use the instantaneous queue size information to decide the switch
configuration. 

Our frame-based offline scheduling algorithm is shown in Algorithm \ref{alg:offline}, which is a modified version of Algorithm 1
in \cite{KimSunMedEryKot:11} and $ \pi_{0}$ in \cite{TasEph:92} using $ (f_{i},u_{j})$-incidence vectors. 

\begin{algorithm}
 \caption{Offline Scheduling Algorithm}
 \label{alg:offline}
   \begin{algorithmic}[1]
 \STATE Consider a traffic pattern with rate vector $\mathbf r$ and its conflict graph $ G$. We assume that $\mathbf r\in STAB( G)$. Then
 \begin{equation}\label{eq:offline}
 r_{i,j}=\sum_{\ell=1}^{m} \phi_\ell \mathcal C^{i,j}(\ell), \quad \forall i,\forall j
 \end{equation}
 where $\sum_{\ell}\phi_\ell=1$ and $\phi_\ell\geq 0, \forall \ell$, $r_{i,j}\geq 0, \forall i, \forall j$. 
 \STATE Assuming all rates $r_{i,j}F$ and $\phi_\ell F$ are rational, choose $F$ such that $r_{i,j}F$ and $\phi_\ell F$ are integers for all $i, j, \ell$.
 \STATE For each $\ell$, use the switch configuration corresponding to $S^\ell$ for $\phi_\ell F$
 slots. If there are fewer than $r_{i,j}F$ requests in the queue, then serve all of
 them. Repeat step 3. 
   \end{algorithmic}
 \end{algorithm}

\begin{thm}
A \ac{QCN} model that follows Algorithm \ref{alg:offline} is
stable if and only if the operating point is within the rate region of the \ac{QCN} model.  
\end{thm}
\begin{proof}
To prove that under the offline algorithm the queues $q_{f_i,u_j}, \forall i, \forall j$
are stable, it is enough to show that the average service rate of queue $q_{f_i,u_j}$ is always
greater that or equal to the arrival rate of flow $(f_i,u_j)$.  
Essentially,  \eqref{eq:offline} expresses the rate $r_{i,j}$ as a convex combination of the stable
sets, which in turn leads to a switch schedule.  This is similar to switch schedules generated via
the Birkhoff-von Neumann \cite{McKMekAnaWal:99} theorem.  From \eqref{eq:offline}, it can be seen that on
average the summation of the fraction of times allocated to each of the stable sets guarantees that
the service time of each queue $q_{f_i,u_j}, \forall i, \forall j$ is at least equal to the arrival
rate of requests to that queue.   
  \end{proof}

We note that there exist online scheduling algorithms, which do not require knowledge of
incoming traffic statistics, which can be applied to our \ac{QCN} model.  For instance,
\cite{TasEph:92} provides an online scheduling policy that achieves
maximum throughput if the arrivals into queues are i.i.d. and independent across incoming flows.
Since this property holds in our \ac{QCN} model, this online policy can be directly applied, in which
case the weight assigned to each vertex in the conflict graph is the sum of all queue lengths of the
ingress queues to which it is associated.

\vspace{-2mm}
\section{The Effect of Coded Storage}
\label{sec:effect-coded-storage}

Given an uncoded \ac{QCN} model, we show how to generate an associated coded \ac{QCN}
model, and then compare the rate regions of these two systems.  Since analyzing the rate region of
coded storage is nontrivial, we develop an upper bound on its rate region instead.  

The upper bound for the coded rate region is generated as follows.  For
each coded chunk $ f_{i}^{c}$ stored on $D_{k}$ (which is a linear combination of chunks $\{ f_{l}
\}_{l \in g(i) } $ in its
generation), from each ingress buffer in the set $ \{ q_{f_{l},u_{j}} \}_{j=1}^{N}$, add a labeled
edge or link to  sink $ q_{u_{j}}$ labeled $D_{k}$.  See Fig.~\ref{fig:QSSBasicToCodedEx} as a
simple single user
example, where in the uncoded physical system two unique chunks are stored on unique drives.  In the coded
physical system, $f_1^{c} = f_{1} + f_{2}$, and $f_2^{c} = a_{1}f_{1} + a_{2}f_{2}$, so we add two edges to the
 coded \ac{QCN} model, as compared to the uncoded \ac{QCN} model.  In the general case, in
 constructing this upper bound we assume that any  request in an ingress buffer is successfully serviced by
sending any coded chunk that has the requested chunk in its generation, and that no penalty is paid
for the user needing to wait to decode the chunk of interest by receiving sufficient degrees of freedom.  

The upper bound is equivalent to setting all elements of the user's file knowledge matrix
$\mathcal{S}$ equal to one, in all states.  This upper bound can be met if coefficient cycling is performed, where
coefficient cycling is the dynamic updating or refreshing of coefficients in a coded chunk such that
every degree of freedom request can be served by any coded chunk, as first defined in
\cite{FerMedSol:12}.  Coefficient cycling guarantees reading of chunks from drives
without replacement.  If chunks are read without replacement within the same generation,
then we can then directly apply Algorithm \ref{alg:offline}.  Similar to \cite{AceDebMedKoe:05},
which considers the achievability of coded storage in systems with an infinite number of storage
nodes, we consider achievability given sufficient storage space and chunk layouts on drives.  Since
it is unclear whether or not coefficient cycling is feasible in all systems, we
show that the coded storage upper bound is achievable if all drives store
sufficient coded information, even without dynamic coding.\footnote{Without loss of
    generality, we assume each coded chunk is unique and therefore, multisets are not required.}
\begin{lem}
  Suppose a \ac{QCN} model operates with a unicast traffic pattern.  The coded storage upper bound is achievable if, for every $ f_{i}^{c}$ stored on drive $ D_k$, then
  drive $D_{k}$ also stores at least $s-1$ additional unique coded chunks from the same generation, $ \{  f_{l}^{c} \colon l\in
  g(i), \, l\neq i \}$.  
\end{lem}
\begin{proof}
The upper bound holds if all ingress buffers can be served by all valid modes across their connected edges, regardless of any user's state via $ \mathcal{S} $.   More specifically,
select any ingress buffer $ q_{u_{j},f_{i}} \in \mathcal{Q} ^{I}$ with connected edges labelled $
  D_{k}$.  For any fixed $(i,j,k)$, if $ \mathcal{S} (i,j,k)=1$ in any timeslot, then $ \mathcal{S} (i,j,k)=1$ must hold in
  all timeslots.  If there exists a read request in $ q_{u_{j},f_{i}}$, then user $ u_{j}$ has received less than $s$ innovative degrees of freedom.  Due to the
  unicast traffic pattern, since $ D_{k}$ stores at least $s$ unique coded chunks, then no matter which strict subset of chunks stored
on $ D_{k}$ $ u_{j}$ has already received, $ D_{k}$ can always serve a coded chunk that is innovative for $
u_{j}$.  Hence, $ u_{j}$ can receive an innovative chunk in every timeslot and if $ \mathcal{S}
(i,j,k)=1$ in any timeslot, then $ \mathcal{S} (i,j,k)=1$ in all timeslots.  
\end{proof}
\begin{lem}
  Suppose a \ac{QCN} model operates with any traffic pattern.  The coded storage upper bound is achievable if, for every $ f_{i}^{c}$ stored on drive $ D_k$, then
  drive $D_{k}$ also stores at least $s+N$ additional unique coded chunks from the same generation $ \{  f_{l}^{c} \colon l\in
  g(i), l\neq i \}$.
\end{lem}
\begin{proof}
  The upper bound holds if all ingress buffers can be served by all their connected edges and valid modes,
  regardless of $ \mathcal{S}$.  Consider the ingress queues $ q_{u_{j},f_{i}} \in \mathcal{Q} ^{I}$ with
  connected edges labeled $D_{k}$.  Suppose there exists a subset of users $ N_{k} \subseteq \{
1, \ldots , N \}$ for which user state  $ \mathcal{S} (i,j,k)=1 , \, \, \forall j \in N_{k}$ in some timeslot,
and $ \exists j\in N_{k}$ such that $ \mathcal{S} (i,j,k)=0$ in another timeslot.  Consider the timeslot in
which $ \exists j\in N_{k}$ such that $ \mathcal{S} (i,j,k)=0$.  This implies $D_{k}$ contains no coded chunk that
is innovative for all users $\{ u_{j} \colon j \in N_{k} \}$.  By definition of $ \mathcal{S}$, all users in $ N_{k}$ are in a state where they have
received less than $s$ innovative coded chunks.  However, the scheduler can then read one of the additional $N$ coded
chunks, since the drive stores $ s+N$ coded chunks from the same generation.  We can continue this process for all timeslots until all users have received $s$
degrees of freedom, so we have a contradiction.
\end{proof}

Note that the uncoded system rate region is a lower bound for the coded storage.   See Sec.~\ref{sec:numerical-results}
for examples that compare the uncoded \ac{RR} with the coded \ac{RR} upper bound.  

\begin{figure}[tb]
  \centering
  \includegraphics[width=\linewidth]{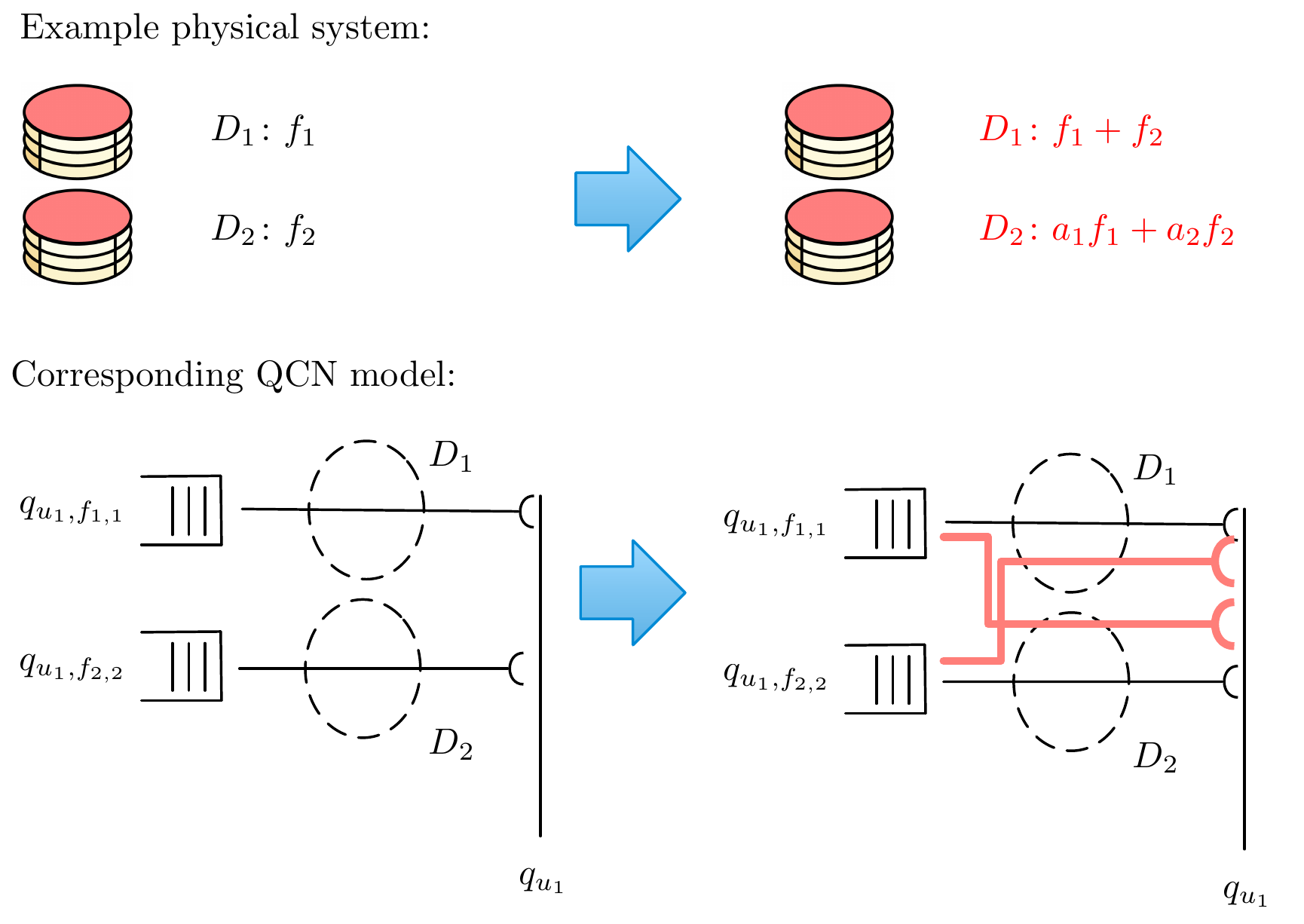}
  \caption{Example of the intuition behind generating, from an uncoded \ac{QCN} model, an upper
    bound for the coded
    \ac{QCN} equivalent.  In the coded physical network example, two chunks are now coded, and so two
    additional links are added into the coded \ac{QCN} model.  Intuitively, coding increases the
    number of links in the \ac{QCN} system, so a scheduler has more scheduling combinations when routing
    requests to drives.}
  \label{fig:QSSBasicToCodedEx}
\end{figure}

\section{Examples \& Numerical Results}
\label{sec:numerical-results}

This section walks the reader through the generation of an uncoded system's \acf{RR} under a variety
of traffic patterns.  All examples are of small storage systems to allow for ease of presentation.
We then compare the uncoded \ac{RR} and coded storage \acp{RR} upper bound.  

\subsection{Uncoded \ac{RR} Examples}
\label{sec:simple-acrr-examples}
This subsection walks the reader through simple \ac{RR} computation examples using our conflict
graph and \ac{SSP} approach.  

\textbf{Ex.~1, one chunk, two users under multicast:} Consider a system with two users $u_1$ and $u_2$ and a single chunk $f_1$ stored
on drive $D_1$, under a multicast traffic pattern. Assuming arrival rate of $r_{i,j}$ for requests of $f_i$ from user
$u_j$, $i=1, j=1,2$, and a general multicast scenario, the conflict graph is shown in
Fig.~\ref{fig:Exp1}.  In this conflict graph, three stable sets can be found, where the incidence
vectors corresponding to these stable sets are as follows
\begin{eqnarray}
\chi^{S_1}=[1,0,0 ], \,\chi^{S_2}=[0,1,0 ], \, \chi^{S_3}=[0,0,1 ] \, .
\end{eqnarray}
The first, second and third elements of these incidence vectors correspond to vertices
$v_{f_{1,1},u_1}$, $v_{f_{1,1},u_2}$, and $v_{f_{1,1},u_{1,2}}$, respectively. Based on the above
incidence vectors, the $(f_{1,1},u_1)-$ and $(f_{1,1},u_2)-$incidence vectors, $\mathcal C^{1,1,1}$
and $ \mathcal C^{1,1,2}$, can be expressed as
\begin{eqnarray}
\mathcal C^{1,1,1}=[1,0,1 ], \, \mathcal C^{1,1,2}=[0,1,1 ] \, .
\end{eqnarray}
Since there exists only one drive in the system, the $(f_1,u_1)$- and $(f_1,u_2)$-incidence vectors, $\mathcal C^{1,1}$ and $\mathcal C^{1,2}$, can be obtained as
\begin{eqnarray}
\mathcal C^{1,1}=\mathcal C^{1,1,1}=[1,0,1 ], \,\mathcal C^{1,2}=\mathcal C^{1,1,2}=[0,1,1 ] \, .
\end{eqnarray}
Using \eqref{eq:offline}, the following system of linear equations can be obtained
\begin{eqnarray}
&\phi_1+\phi_3=r_{1,1}, \, \phi_2+\phi_3=r_{1,2}, \, \phi_1+\phi_2+\phi_3=1 \, \nonumber\\
&r_{1,1}\geq 0,r_{1,2}\geq 0,\, \phi_1\geq 0, \phi_2\geq 0, \phi_3\geq 0 \, . 
\end{eqnarray}
The \ac{RR} corresponding to this system of linear equations is shown in Fig.~\ref{fig:Exp1}(b), with area of 1.

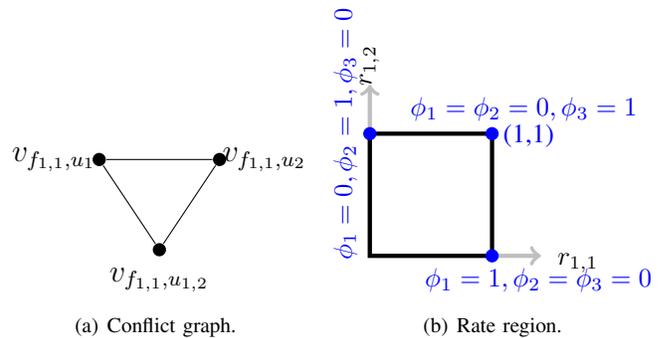
\begin{figure}
\centering
\subfigure[Conflict graph. ]{\begin{tikzpicture}[scale=0.8]
 \node[vertex,minimum size=5pt,color=black,text=black] (v1) at (-1,1.5)  {};
 \node[vertex,minimum size=5pt,color=black,text=black] (v2) at (1,1.5)  {};
 \node[vertex,minimum size=5pt,color=black,text=black] (v3) at (0,0)  {};
\node [font=\large] at (-1.75,1.5) {$v_{f_{1,1},u_1}$};
\node [font=\large] at (1.75,1.5) {$v_{f_{1,1},u_2}$};
\node [font=\large] at (0,-0.5) {$v_{f_{1,1},u_{1,2}}$};
 \draw[-] (v1) to (v3);
 \draw[-] (v2) to (v3);
 \draw[-] (v1) to (v2);
 \end{tikzpicture}}
\subfigure[Rate region.]{\begin{tikzpicture}[scale=0.65]
\draw [ultra thick, color=black!25, <->] (0,3.5) -- (0,0) -- (3.5,0);
\node [below right] at (3.65,0.25) {$r_{1,1}$};
\node [left,rotate=90] at (0,4.5) {$r_{1,2}$};
\draw [ultra thick,-] (0,2.5) -- (2.5,2.5) -- (2.5,0) -- (0,0) -- (0,2.5);
\tikzstyle{vertex}=[auto=left,circle,fill=blue!25,minimum size=20pt,inner sep=0pt]
\node[vertex,minimum size=5pt,color=blue] (n1) at (2.5,0) {};
\node[vertex,minimum size=5pt,color=blue] (n1) at (0,2.5) {};
\node[vertex,minimum size=5pt,color=blue] (n1) at (2.5,2.5) {};
\node [rotate=90,text=black,color=blue!] at (-0.5,2.5) {$\phi_1=0,\phi_2=1,\phi_3=0$};
\node [rotate=0,text=black,color=blue!] at (3.45,-0.5) {$\phi_1=1,\phi_2=\phi_3=0$};
\node [rotate=0,text=black,color=blue!] at (3.15,3) {$\phi_1=\phi_2=0,\phi_3=1$};
\node [rotate=0,text=black,color=blue!] at (3.25,2.45) {(1,1)};
\end{tikzpicture}}
\caption{Conflict graph and \ac{RR} for Ex.~1, with two users and a single chunk, operating with a multicast
  traffic pattern.}
\label{fig:Exp1}
\end{figure}

 \textbf{Ex.~2, one chunk, two users under unicast}: Consider Ex.~1 under the unicast traffic pattern, as
 opposed to multicast. The updated conflict graph is  shown in Fig.~\ref{fig:Exp2}(a).  In this
 conflict graph, two stable sets can  be found and the incidence vectors corresponding to these
 stable sets are as follows
\begin{eqnarray}
\chi^{S_1}=[1,0 ], \, \chi^{S_2}=[0,1 ] \, .
\end{eqnarray}
The first and second elements of these incidence vectors correspond to vertices $v_{f_{1,1},u_1}$ and
$v_{f_{1,1},u_2}$, respectively.  Based on the above incidence vectors and the fact that there exists only one drive in the system, the $(f_1,u_1)$- and $(f_1,u_2)$-incidence vectors, $\mathcal C^{1,1}$ and $\mathcal C^{1,2}$, can be expressed as
\begin{eqnarray}
\mathcal C^{1,1}=\mathcal C^{1,1,1}=[1,0 ],\,  \,\mathcal C^{1,2}=\mathcal C^{1,2,1}=[0,1 ] \, .
\end{eqnarray}
By using \eqref{eq:offline}, the following system of linear equations can be obtained
\begin{eqnarray}
&\phi_1=r_{1,1}, \, \phi_2=r_{1,2}, \,\phi_1+\phi_2=1 \, \nonumber\\
&r_{1,1}\geq 0,r_{1,2}\geq 0, \, \phi_1\geq 0, \phi_2\geq 0 \, . 
\end{eqnarray}
The corresponding rate region is shown in Fig.~\ref{fig:Exp2}(b), with area of 1/2. 

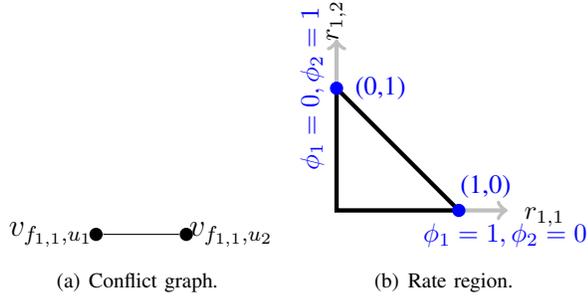
\begin{figure}
\centering
\subfigure[Conflict graph. ]{\begin{tikzpicture}[scale=0.8]
 \node[vertex,minimum size=5pt,color=black,text=black] (v1) at (-0.5,3.5)  {};
 \node[vertex,minimum size=5pt,color=black,text=black] (v2) at (1,3.5)  {};
\node [font=\large] at (-1.25,3.5) {$v_{f_{1,1},u_1}$};
\node [font=\large] at (1.75,3.5) {$v_{f_{1,1},u_2}$};
 \draw[-] (v1) to (v2);
 \end{tikzpicture}}
\subfigure[ Rate region.]{\begin{tikzpicture}[scale=0.65]
\draw [ultra thick, color=black!25, <->] (0,3.5) -- (0,0) -- (3.5,0);
\node [below right] at (3.65,0.25) {$r_{1,1}$};
\node [left,rotate=90] at (0,4.5) {$r_{1,2}$};
\draw [ultra thick,-] (0,2.5) -- (2.5,0) -- (0,0) -- (0,2.5);
\tikzstyle{vertex}=[auto=left,circle,fill=blue!25,minimum size=20pt,inner sep=0pt]
\node[vertex,minimum size=5pt,color=blue] (n2) at (2.5,0) {};
\node[vertex,minimum size=5pt,color=blue] (n3) at (0,2.5) {};
\node [rotate=90,text=black,color=blue!] at (-0.5,2.5) {$\phi_1=0,\phi_2=1$};
\node [rotate=0,text=black,color=blue!] at (3.45,-0.5) {$\phi_1=1,\phi_2=0$};
\node [rotate=0,text=black,color=blue!] at (0.9,2.5) {(0,1)};
\node [rotate=0,text=black,color=blue!] at (3.05,0.5) {(1,0)};
\end{tikzpicture}}
\caption{Conflict graph and \ac{RR} for Ex.~2, with two users and a single chunk, operating with a unicast
  traffic pattern.}
\label{fig:Exp2}
\end{figure}

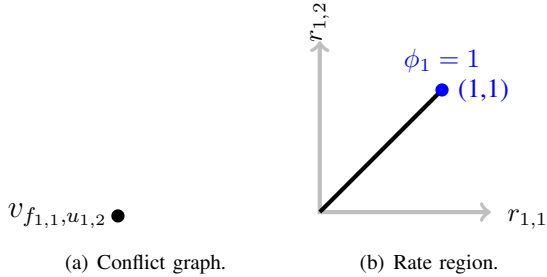
\begin{figure}
\centering
\subfigure[ Conflict graph.]{\begin{tikzpicture}[scale=0.8]
 \node[vertex,minimum size=5pt,color=black,text=black] (v1) at (-3.25,0.5)  {};
\node [font=\large] at (-4.25,0.5) {$v_{f_{1,1},u_{1,2}}$};
\node [font=\large] at (-0.5,0.5) {};
 \end{tikzpicture}}
\subfigure[Rate region.]{\begin{tikzpicture}[scale=0.65]
\draw [ultra thick, color=black!25, <->] (0,3.5) -- (0,0) -- (3.5,0);
\node [below right] at (3.65,0.25) {$r_{1,1}$};
\node [left,rotate=90] at (0,4.5) {$r_{1,2}$};
\draw [ultra thick,-]  (0,0) -- (2.5,2.5);
\tikzstyle{vertex}=[auto=left,circle,fill=blue!25,minimum size=20pt,inner sep=0pt]
\node[vertex,minimum size=5pt,color=blue] (n2) at (2.5,2.5) {};
\node [rotate=0,text=black,color=blue!] at (3.35,2.5) {(1,1)};
\node [rotate=0,text=black,color=blue!] at ((2.5,3.1) {$\phi_1=1$};
\end{tikzpicture}}
\caption{Conflict graph and \ac{RR} for Ex.~3, with two users and a single chunk, operating with a broadcast
  traffic pattern.}
\label{fig:Exp3}
\end{figure}

 \textbf{Ex.~3, one chunk, two users under broadcast}: Consider Ex.~ 1 under the broadcast traffic pattern.
 This example shows that there are caveats in using \ac{RR} area/volume as a performance metric that should
 be carefully considered.  The conflict graph is shown in Fig.~\ref{fig:Exp3}(a),
 in which only one stable set exists,  with incidence vector
\begin{eqnarray}
\chi^{S_1}&=&[1] \, .
\end{eqnarray}
Based on the above incidence vector, the $(f_1,u_1)$- and $(f_1,u_2)$-incidence vectors, $\mathcal C^{1,1}$ and
$\mathcal C^{1,2}$, are
\begin{eqnarray}
\mathcal C^{1,1}=\mathcal C^{1,1,1}=[1], \,\mathcal C^{1,2}=\mathcal C^{1,1,2}=[1] \, .
\end{eqnarray}
By using \eqref{eq:offline}, the following system of linear equations can be obtained
\begin{eqnarray}
&\phi_1=r_{1,1}, \,\phi_1=r_{1,2}, \, \phi_1=1 \nonumber\\
&r_{1,1}\geq 0,r_{1,2}\geq 0, \, \phi_1\geq 0 \, . 
\end{eqnarray}
The corresponding \ac{RR}  is shown in Fig.~\ref{fig:Exp3}(b), with area of 0.

\subsection{Comparison of Uncoded and Coded Storage \acp{RR}}
\label{sec:comp-uncod-coded}

This subsection compares the \ac{RR} areas of uncoded storage and the coded storage upper bound.
Our coded storage numerical examples use a striped file coded storage layout, which can be seen in \cite{FerMedSol:12},
and assume $ s/T \in \mathbb N_{+}$.  We first walk the reader through a very simple example that
demonstrates the increased scheduling flexibility allowed by coded storage.  We then summarize other
examples across various traffic patterns.  

\textbf{Ex.~4, uncoded, two chunks, one user with $\mathbf {R_{x}(1) = 2}$:}  Consider an uncoded system with one user $u_1$ and two chunks
$f_1$ and $f_2$ which are stored on drives $D_1$ and $D_2$, respectively.  Assuming arrival rate of
$r_{i,j}$ for request of $f_i$ from user $u_j$, $i=1,2, j=1$, and a
multicast with multipacket reception setting, the conflict graph is shown in
Fig~\ref{fig:Exp_two_one_uc_MU}(a).   In this conflict graph, three stable sets can be found, where the incidence vectors corresponding to these stable sets are as follows
\begin{eqnarray}
\chi^{S_1}=[1,0 ], \,\chi^{S_2}=[0,1 ], \,\chi^{S_3}=[1,1 ]\, .
\end{eqnarray}
The first and second elements of these incidence vectors correspond to vertices $v_{f_{1,1},u_1}$ and
$v_{f_{2,2},u_1}$, respectively.  Based on the above incidence vectors, the $(f_1,u_1)$- and
$(f_2,u_1)$-incidence vectors, $\mathcal C^{1,1}$ and $\mathcal C^{2,1}$, can be expressed  as
\begin{eqnarray}
\mathcal C^{1,1}=\mathcal C^{1,1,1}= [1,0,1 ],  \,\mathcal C^{2,1}=\mathcal C^{2,1,2}= [0,1,1 ] \, .
\end{eqnarray}
By using \eqref{eq:offline}, the following system of linear equations can be obtained
\begin{eqnarray}
&\phi_1+\phi_3=r_{1,1},\, \phi_2+\phi_3=r_{2,1},\,\phi_1+\phi_2+\phi_3=1\nonumber\\
&r_{1,1}\geq 0,r_{2,1}\geq 0,\, \phi_1\geq 0, \phi_2\geq 0, \phi_3\geq 0 \, . 
\end{eqnarray}
The \ac{RR} corresponding to the above system of linear equations is shown in Fig.~\ref{fig:Exp_two_one_uc_MU}(b), with area 1.

\begin{figure}
\centering
\subfigure[Uncoded conflict graph.]{\begin{tikzpicture}[scale=.8]
 \node[vertex,minimum size=5pt,color=black,text=black] (v1) at (-1,1.5)  {};
 \node[vertex,minimum size=5pt,color=black,text=black] (v2) at (1,1.5)  {};
\node [font=\large] at (-1.75,1.5) {$v_{f_{1,1},u_1}$};
\node [font=\large] at (1.75,1.5) {$v_{f_{2,2},u_1}$};
 \end{tikzpicture}}
\subfigure[Uncoded rate region.]{\begin{tikzpicture}[scale=0.65]
\draw [ultra thick, color=black!25, <->] (0,3.5) -- (0,0) -- (3.5,0);
\node [below right] at (3.65,0.25) {$r_{1,1}$};
\node [left,rotate=90] at (0,4.5) {$r_{2,1}$};
\draw [ultra thick,-] (0,2.5) -- (2.5,2.5) -- (2.5,0) -- (0,0) -- (0,2.5);
\tikzstyle{vertex}=[auto=left,circle,fill=blue!25,minimum size=20pt,inner sep=0pt]
\node[vertex,minimum size=5pt,color=blue] (n1) at (2.5,0) {};
\node[vertex,minimum size=5pt,color=blue] (n1) at (0,2.5) {};
\node[vertex,minimum size=5pt,color=blue] (n1) at (2.5,2.5) {};
\node [rotate=90,text=black,color=blue!] at (-0.5,2.5) {$\phi_1=0,\phi_2=1,\phi_3=0$};
\node [rotate=0,text=black,color=blue!] at (3.4,-0.5) {$\phi_1=1,\phi_2=\phi_3=0$};
\node [rotate=0,text=black,color=blue!] at (3.25,3) {$\phi_1=\phi_2=0,\phi_3=1$};
\node [rotate=0,text=black,color=blue!] at (3.15,2.45) {(1,1)};
\end{tikzpicture}}
\caption{Conflict graph and \ac{RR} for uncoded Ex.~4, with one user, two chunks, and multi-packet
  reception, $R_{x}(1)=2$.}
\label{fig:Exp_two_one_uc_MU}
\end{figure}
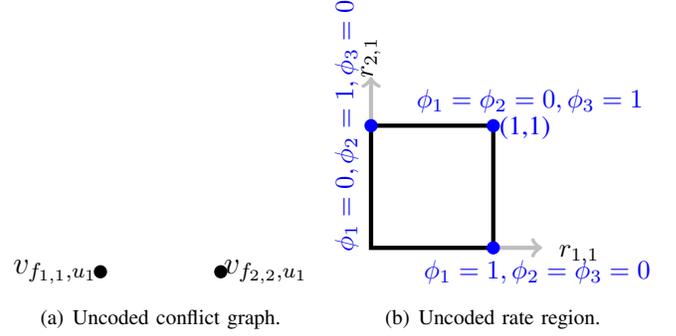

\textbf{Ex.~5, coded, two chunks, one user with $\mathbf {R_{x}(1) = 2}$:} Now consider Ex.~4 under a coded storage
system, where $D_1:f_1+f_2$ and $D_2:a_1f_1+a_2f_2$.   The conflict graph of the coded system is
shown in Fig~\ref{fig:Exp_two_one_c_MU}(a).   In this conflict graph, eight stable sets can be
found, where the incidence vectors corresponding to these stable sets are as follows
\begin{eqnarray}
\chi^{S_1}=[1,0,0,0 ], \,\chi^{S_2}=[0,1,0,0 ]\, \, \,\nonumber\\
\chi^{S_3}=[0,0,1,0 ],\,\chi^{S_4}=[0,0,0,1 ]\, \,  \,\nonumber\\
\chi^{S_5}=[1,1,0,0 ],\,\chi^{S_6}=[1,0,0,1 ]\, \, \, \nonumber\\
\chi^{S_7}=[0,1,1,0 ],\,\chi^{S_8}=[0,0,1,1 ]\, .
\end{eqnarray}
The first, second, third and fourth elements of these incidence vectors correspond to vertices
$v_{f_{1,1},u_1}$, $v_{f_{1,2},u_1}$, $v_{f_{2,1},u_1}$ and $v_{f_{2,2},u_1}$, respectively.
Furthermore, based on the above incidence vectors, the $(f_{1,1},u_1)-$ , $(f_{1,2},u_1)-$,
$(f_{2,1},u_1)-$ and $(f_{2,2},u_1)$-incidence vectors, $\mathcal C^{1,1,1}$, $\mathcal
C^{1,2,1}$,  $\mathcal C^{2,1,1}$ and $\mathcal C^{2,2,1}$, can be expressed as
\begin{eqnarray}
\mathcal C^{1,1,1}=[1,0,0,0,1,1,0,0 ],\,\mathcal C^{1,2,1}=[0,1,0,0,1,0,1,0 ]\,\,\,\nonumber\\
\mathcal C^{2,1,1}=[0,0,1,0,0,0,1,1 ],\,\mathcal C^{2,2,1}=[0,0,0,1,0,1,0,1 ] \, .
\end{eqnarray}
Therefore, the $(f_1,u_1)$- and $(f_2,u_1)$-incidence vectors, $\mathcal C^{1,1}$ and $\mathcal C^{2,1}$, are
\begin{eqnarray}
\mathcal C^{1,1}=\mathcal C^{1,1,1}+\mathcal C^{1,2,1}\,=[1,1,0,0,2,1,1,0]\,  \,\nonumber\\
\mathcal C^{2,1}=\mathcal C^{2,1,1}+\mathcal C^{2,2,1}=[0,0,1,1,0,1,1,2 ] \, .
\end{eqnarray}
Then, by using \eqref{eq:offline}, the following system of linear equations can be obtained
\begin{eqnarray}
&\phi_1+\phi_2+2\phi_5+\phi_6+\phi_7=r_{1,1}\nonumber\\
&\phi_3+\phi_4+\phi_6+\phi_7+2\phi_8=r_{2,1}\nonumber\\
&\phi_1+\phi_2+\phi_3+\phi_4+\phi_5+\phi_6+\phi_7+\phi_8=1\nonumber\\
&r_{1,1}\geq 0,r_{2,1}\geq 0,\, \phi_1\geq 0, \phi_2\geq 0, \phi_3\geq 0, \phi_4\geq 0, \nonumber\\ & \phi_5\geq 0, \phi_6\geq 0, \phi_7\geq 0,
\phi_8\geq 0 \, . 
\end{eqnarray}
The corresponding rate region can be obtained as shown in Fig.~\ref{fig:Exp_two_one_c_MU}(b), of area
2, whereas the uncoded equivalent has area of 1 (as shown in Fig.~\ref{fig:Exp_two_one_uc_MU}(b)).  

Intuition behind this increase is as follows.  Suppose two requests are in queue $
q_{f_{1,u_{1}}}$ and none are in $ q_{f_{2},u_{1}}$.   In this case, in the uncoded system servicing
this would take two timeslots.  In comparison, the coded system could service these two requests in
one timeslot.

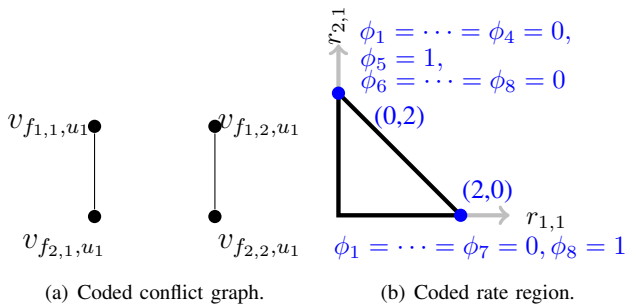
\begin{figure}[thb!]
\centering
\subfigure[Coded conflict graph.]{\begin{tikzpicture}[scale=0.8]
\node[vertex,minimum size=5pt,color=black,text=black] (v1) at (-1,1.5)  {};
\node[vertex,minimum size=5pt,color=black,text=black] (v2) at (1,1.5)  {};
\node[vertex,minimum size=5pt,color=black,text=black] (v3) at (-1,0)  {};
\node[vertex,minimum size=5pt,color=black,text=black] (v4) at (1,0)  {};
\node [font=\large] at (-1.75,1.5) {$v_{f_{1,1},u_1}$};
\node [font=\large] at (1.75,1.5) {$v_{f_{1,2},u_1}$};
\node [font=\large] at (-1.5,-0.5) {$v_{f_{2,1},u_1}$};
\node [font=\large] at (1.75,-0.5) {$v_{f_{2,2},u_{1}}$};
\draw[-] (v1) to (v3);
\draw[-] (v2) to (v4);
\end{tikzpicture}}
\subfigure[Coded rate region.]{\begin{tikzpicture}[scale=0.65]
\draw [ultra thick, color=black!25, <->] (0,3.5) -- (0,0) -- (3.5,0);
\node [below right] at (3.65,0.25) {$r_{1,1}$};
\node [left,rotate=90] at (0,4.5) {$r_{2,1}$};
\draw [ultra thick,-] (0,2.5) -- (2.5,0) -- (0,0) -- (0,2.5);
\tikzstyle{vertex}=[auto=left,circle,fill=blue!25,minimum size=20pt,inner sep=0pt]
\node[vertex,minimum size=5pt,color=blue] (n2) at (2.5,0) {};
\node[vertex,minimum size=5pt,color=blue] (n3) at (0,2.5) {};
\node [rotate=0,text=black,color=blue!] at (1.25,2) {(0,2)};
\node [rotate=0,text=black,color=blue!] at (3.05,0.5) {(2,0)};
\node [rotate=0,text=black,color=blue!] at (2.6,3.75) {$\phi_1=\cdots=\phi_4=0,$};
\node [rotate=0,text=black,color=blue!] at (1.3,3.2) {$\phi_5=1,$};
\node [rotate=0,text=black,color=blue!] at (2.55,2.75) {$\phi_6=\cdots=\phi_8=0$};
\node [rotate=0,text=black,color=blue!] at (2.9,-0.65) {$\phi_1=\cdots=\phi_7=0,\phi_8=1$};
\end{tikzpicture}}
\caption{Conflict graph and \ac{RR} for coded Ex.~5, with one user, two chunks, and multi-packet
  reception, $R_{x}(1)=2$.}
\label{fig:Exp_two_one_c_MU}
\end{figure}

\textbf{Ex.~6, uncoded and coded, two users, two chunks}: Consider an uncoded system with two users $u_1$, $u_2$ and two chunks $f_1$ and $f_2$ stored on drives $D_1$ and $D_2$. 
For the coded system, consider the following drive mapping: $D_1:f_1+f_2$ and $D_2:a_1f_1+a_2f_2$. A
comparison of the \ac{RR} volumes under different traffic patterns and MPR assumptions is presented in
Table~\ref{tab:comp1}. Furthermore, as an example, the conflict graphs for the uncoded and coded
systems under a multicast traffic pattern with $R_{x}(1)=R_{x}(2)=1$ are illustrated in
Fig.~\ref{fig:conflict_graph_ex}.  Note neither graph is a simple clique.  For instance, since
multicast is allowed, in the uncoded conflict graph, there are no edges between vertices reading
from different drives and transmitting to different users, e.g. between $v_{f_{1,1},u_1}$ and
$v_{f_{2,2},u_2}$, and between $v_{f_{1,1},u_2}$ and $v_{f_{2,2},u_1}$. Similarly, for the coded system,
there are no edges between vertices of the form  $v_{f_{i,k},u_j}$ and $v_{f_{i,k'},u_{j'}}$,
$i,k,j\in\{1,2\}$, where $k\neq k'$ and $j\neq j'$.           

\begin{table}[tb]
\caption{ Comparison of \ac{RR} volumes for Ex.~6, 2 chunks, 2 users, under uncoded and coded storage.}
\vspace{-5mm}
\centering
\bigskip
\ra{1.3}
\begin{tabular}{@{}l  r r r@{}}
\toprule
\textbf{Traffic Pattern} & \textbf{Uncoded} & \textbf{Coded} & \textbf{\% $\Delta$}\\
\midrule
Single unicast & 0.0417 & 0.0417 & 0\\
Multiple unicast & 0.1667 & 0.25 & 50\\
Multiple Unicast, with MPR & 0.25 & 0.6667 & 167\\
Broadcast& 0 & 0 & 0\\
Broadcast, with MPR & 0 & 0 & 0\\
Multicast& 0.25 & 0.25 & 0\\
Multicast, with MPR  & 1 & 2.6667 & 167\\
\midrule
\textbf{Average} & & & 54.8 \\
\bottomrule
\end{tabular}
\label{tab:comp1}
\end{table}



\begin{figure}[tb]
  \centering
  \subfigure[Uncoded conflict graph.]  {\includegraphics[width=0.6\linewidth]{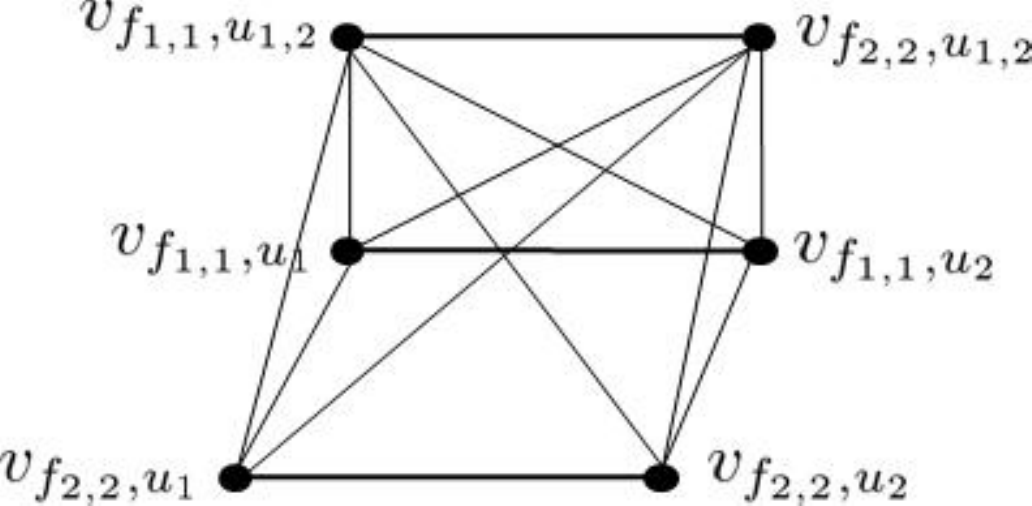}}
\subfigure[Coded conflict graph.] {\includegraphics[width=\linewidth]{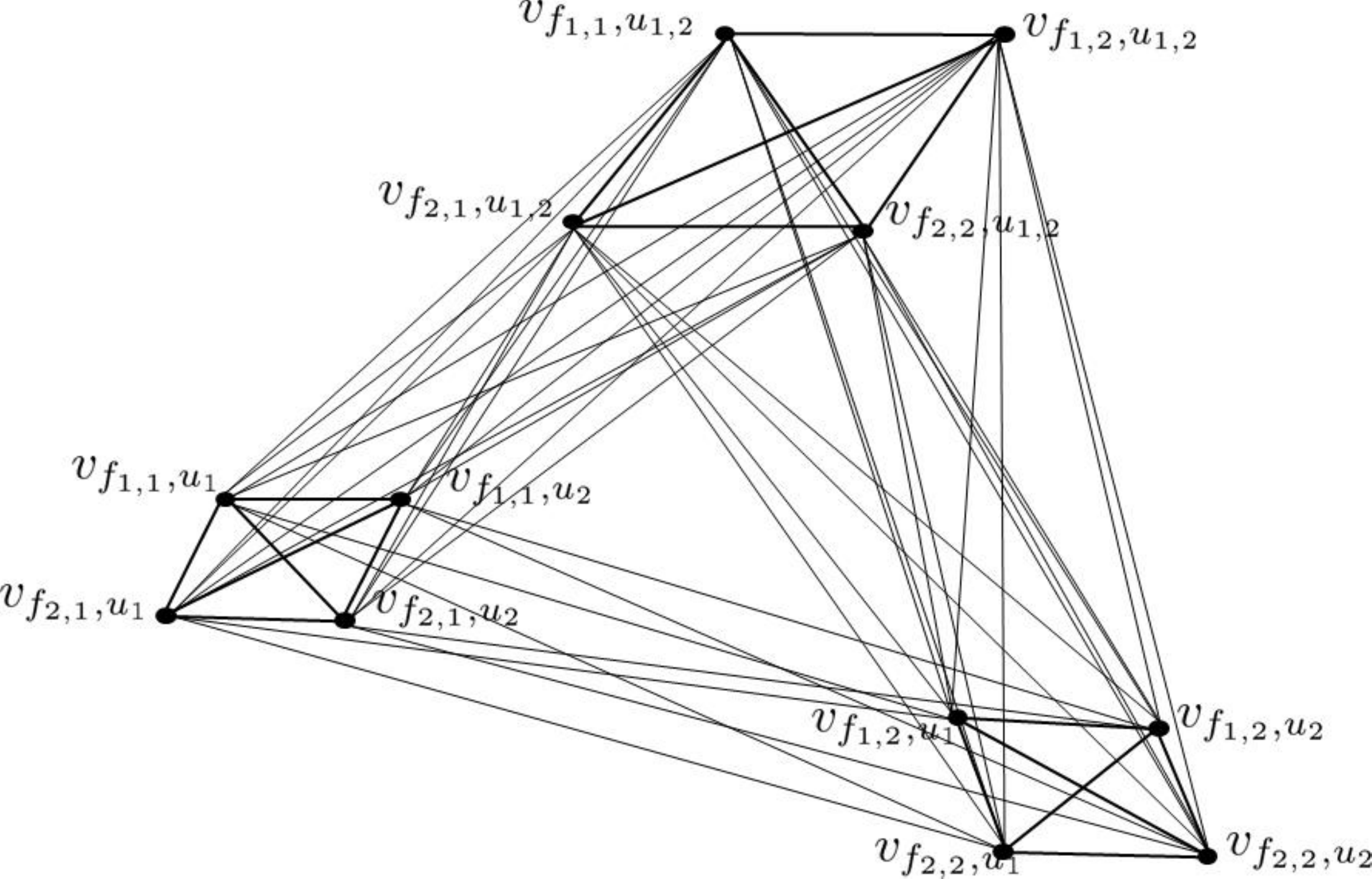}}
  \caption{Conflict graph for uncoded and coded system with two chunks and two users under the multicast setting with $R_{x}(1)=R_{x}(2)=1$. }
  \label{fig:conflict_graph_ex}
\end{figure}

\textbf{Ex.~7, uncoded and coded, two users, three chunks:} Consider an uncoded system
with two users $u_1$, $u_2$ and three chunks $f_1$, $f_2$, and $f_3$ stored on drives $D_1$, $D_2$
and $D_3$.  
For the coded system, consider the following drive mapping: $D_1:f_1+f_2+f_3$,
$D_2:a_1f_1+a_2f_2+a_3f_3$ and $D_3:a_4f_1+a_5f_2+a_6f_3$, a comparison of the \ac{RR}'s volume for
this system under different traffic patterns 
is presented in Table~\ref{tab:comp2}. 

\begin{table}[tb]
\caption{ Comparison of \ac{RR} volumes for Ex.~7, 2 users, 3 chunks with uncoded storage and the
  coded storage upper bound.}
\vspace{-5mm}
\centering
\bigskip
\ra{1.3}
\begin{tabular}{@{}l  r r r@{}}
\toprule
\textbf{Traffic Pattern} & \textbf{Uncoded} & \textbf{Coded} & \textbf{\% $\Delta$}\\
\midrule
Single unicast & 0.0014& 0.0014 & 0\\
Multiple unicast & 0.0236 & 0.0278 & 17.8\\
Multiple unicast, with MPR & 0.1250 & 1.0125 & 710\\
Broadcast & 0 & 0 & 0\\
Broadcast, with MPR & 0 & 0 & 0\\
Multicast & 0.0278 & 0.0278 & 0\\
Multicast, with MPR  & 1 & 8.1 & 710\\
\midrule
\textbf{Average} & & & 205.4 \\
\bottomrule
\end{tabular}
\label{tab:comp2}
\end{table}
%
%

Results show significant increases in \ac{RR} volume when using coded storage, averaged
across traffic patterns, and as traffic patterns change, the bound shows sizable variability in
coded storage gains.  This encouraging gain is tempered by the fact that it is an upper bound.  Yet,
the size of such potential increases warrants further and more exact coded storage \ac{RR}
analysis.

\vspace{-3mm}
\section{Discussion \& Conclusions}
\label{sec:disc--concl}

Potential areas of future work in this area are as follows.    Although the \ac{QCN} model 
can be used for arbitrary chunk-to-drive mappings, the seeking within drives is assumed to be
deterministic.  This may be a reasonable model if drives are of particular solid state varieties,
but deterministic drive models for \acp{HDD} are problematic.  It would be interesting to extend the
model to allow for internal finite buffers at drives themselves, as well as arbitrary service
distributions.  

The conflict graphs generated by the \ac{QCN} model are a
powerful tool for exploring different traffic patterns.  However, the state space of the conflict
graphs grows quickly in the general case, and it may be useful to explore special cases of storage
and traffic patterns that allow for conflict graphs that may not require particular structures that
allow for exact characterization, such as quasi-line conflict graphs.  

We have developed an upper bound for coded storage rate regions, which is achievable in certain
dynamic coding systems, as well as in certain chunk-to-drive mappings.  It would be useful to
examine the tightness of this bound, as well as to develop a strict lower bound,
or to find an exact solution to the rate region of coded storage.  Potential avenues include
characterizing coded \ac{QCN} models via user's file knowledge matrix $ \mathcal{S} $ directly.  

In conclusion, we have developed a method to map a physical storage system into a simple queued
cross-bar network model, with particular application to high-traffic storage systems.   In doing so,
our method and related analysis tools use existing work in  arbitrary queueing
networks literature, cross-bar switching, as well as conflict graphs.  This allows the \ac{QCN}
method as a natural modeling and analysis tool for systems  with non-regular chunk-to-drive mappings,
replication, as well as for coded storage.   We have used a conflict graph approach, which is a function of
storage system traffic patterns, to exactly characterize the stable set polytope of the conflict
graphs in a number of cases.  We have then computed and compared the rate regions of uncoded storage and
 of the coded storage upper bound, quantifying promising benefits of coded storage over uncoded systems in terms of
\ac{RR} volume.  We have also shown how optimal offline and online scheduling algorithms can be
generated from our model.  

\vspace{-3mm}

\end{document}